\documentclass[a4paper,11pt]{article}
\pdfoutput=1 
\usepackage{jheppub} 
\usepackage{amsmath,amsfonts,amssymb,amsthm}
\usepackage[british]{babel}
\usepackage{bbm,amscd,bbold}
\usepackage{array}
\usepackage[svgnames,table]{xcolor}
\usepackage[T1]{fontenc}
\usepackage[utf8x]{inputenc}
\usepackage{lmodern} 
\usepackage{mathrsfs}
\usepackage{verbatim}
\usepackage{pdfsync}
\usepackage[final]{microtype}
\usepackage{adjustbox}
\usepackage{booktabs} 
\usepackage{caption,enumitem}
\usepackage{appendix}
\usepackage{chngcntr}
\usepackage{apptools}
\usepackage{tikz,pgfplots}
\usepgfplotslibrary{fillbetween}
\usetikzlibrary{calc,arrows,arrows.meta,cd,shapes.geometric,positioning,through,decorations.markings,decorations.text}

\hypersetup{%
  pdftitle   = {Carroll/fracton particles and their correspondence},
  pdfkeywords = {Carroll, coadjoint orbits, particles,fractons,dynamics},
  pdfauthor  = {José Figueroa-O'Farrill, Alfredo Pérez, Stefan Prohazka},
  pdfcreator = {\LaTeX\ with package \flqq hyperref\frqq},
  linkcolor=NavyBlue,
  citecolor=ForestGreen,
  urlcolor=OrangeRed,
  anchorcolor=OrangeRed,
  colorlinks=true
}
\usepackage[hyperpageref]{backref}
\renewcommand*{\backref}[1]{}
\renewcommand*{\backrefalt}[4]{%
  \ifcase #1%
  \or [Page~#2.]%
  \else [Pages~#2.]%
  \fi%
}
\theoremstyle{plain}
\newtheorem{lemma}{Lemma}

\theoremstyle{definition}

\newcommand{\g}{\mathfrak{g}}

\newcommand{\fk}{\mathfrak{k}}

\newcommand{\gl}{\mathfrak{gl}}

\newcommand{\so}{\mathfrak{so}}
\newcommand{\iso}{\mathfrak{iso}}

\renewcommand{\t}{\mathfrak{t}}

\newcommand{\be}{\boldsymbol{e}}
\newcommand{\bu}{\boldsymbol{u}}
\newcommand{\x}{\boldsymbol{x}}

\newcommand{\bj}{\boldsymbol{j}}
\newcommand{\p}{\boldsymbol{p}}
\renewcommand{\k}{\boldsymbol{k}}
\newcommand{\ba}{\boldsymbol{a}}
\newcommand{\bb}{\boldsymbol{b}}

\newcommand{\bzero}{\boldsymbol{0}}

\newcommand{\eO}{\mathcal{O}}

\newcommand{\ad}{\operatorname{ad}}
\newcommand{\Ort}{\operatorname{O}}

\newcommand{\Ad}{\operatorname{Ad}}

\newcommand{\Tr}{\operatorname{Tr}}

\renewcommand{\AA}{\mathbb{A}}

\newcommand{\RR}{\mathbb{R}}

\newcommand{\Aff}{\operatorname{Aff}}

\newcommand{\GL}{\operatorname{GL}}
\newcommand{\ISO}{\operatorname{ISO}}
\newcommand{\SO}{\operatorname{SO}}

\definecolor{dkgr}{rgb}{0,0.6,0}
\definecolor{gris}{rgb}{0.5,0.5,0.5}
\newcommand{\zero}{{\color{gris}0}}
\usepackage{pdflscape}

\usepackage{braket} 
\usepackage{bm} 

\providecommand*{\xb}{{\bm{x}}}
\renewcommand*{\xb}{{\bm{x}}}

\providecommand*{\pd}{\partial}
\renewcommand*{\pd}{\partial}

\providecommand*{\db}{{\bm{d}}}
\renewcommand*{\db}{{\bm{d}}}

\providecommand*{\pb}{{\bm{p}}}
\renewcommand*{\pb}{{\bm{p}}}
\providecommand*{\vb}{{\bm{v}}}
\renewcommand*{\vb}{{\bm{v}}}
\providecommand*{\ab}{{\bm{a}}}
\renewcommand*{\ab}{{\bm{a}}}
\providecommand*{\jb}{{\bm{j}}}
\renewcommand*{\jb}{{\bm{j}}}
\providecommand*{\spb}{{\bm{S}}}
\renewcommand*{\spb}{{\bm{S}}}

\providecommand*{\varphib}{{\bm{\varphi}}}
\renewcommand*{\varphib}{{\bm{\varphi}}}

\providecommand*{\zerob}{{\bm{0}}}
\renewcommand*{\zerob}{{\bm{0}}}

\providecommand*{\xb}{{\bm{x}}}
\renewcommand*{\xb}{{\bm{x}}}

\newcommand{\bv}{\boldsymbol{v}}

\title{\boldmath Carroll/fracton particles and their correspondence}

 \author[a,1]{José Figueroa-O'Farrill,\note{ORCID: \href{https://orcid.org/0000-0002-9308-9360}{0000-0002-9308-9360}}}
 \author[b,c,2]{Alfredo Pérez\note{ORCID: \href{https://orcid.org/0000-0003-0989-9959}{0000-0003-0989-9959}}}
 \author[a,3]{and Stefan Prohazka\note{ORCID: \href{https://orcid.org/0000-0002-3925-3983}{0000-0002-3925-3983}}}

\affiliation[a]{Maxwell Institute and School of Mathematics, The University
  of Edinburgh, James Clerk Maxwell Building, Peter Guthrie Tait Road,
  Edinburgh EH9 3FD, Scotland, United Kingdom}
\affiliation[b]{Centro de Estudios Científicos (CECs), Avenida Arturo Prat 514, Valdivia, Chile}
\affiliation[c]{Facultad de Ingeniería, Arquitectura y Diseño, Universidad San Sebastián, sede Valdivia, General Lagos 1163, Valdivia 5110693, Chile}

 \emailAdd{j.m.figueroa@ed.ac.uk}
 \emailAdd{alfredo.perez@uss.cl}
 \emailAdd{stefan.prohazka@ed.ac.uk}

 \abstract{We exploit the close relationship between the Carroll and
   fracton/dipole algebras, together with the method of coadjoint orbits, to
   define and classify classical Carroll and fracton particles.  This
   approach establishes a Carroll/fracton correspondence and provides
   an answer to the question ``What is a fracton?''.

   Under this correspondence, carrollian energy and center-of-mass
   correspond to the fracton electric charge and dipole moment,
   respectively. Then immobile massive Carroll particles correspond to the
   fracton monopoles, whereas certain mobile Carroll particles
   (``centrons'') correspond to fracton elementary dipoles. We uncover
   various new massless carrollian/neutral fractonic particles,
   provide an action in each case and relate them via a
   $GL(2,\mathbb{R})$ symmetry.

   We also comment on the limit from Poincaré particles, the relation
   to (electric and magnetic) Carroll field theories, contrast Carroll
   boosts with dipole transformations and highlight a generalisation
   to curved space ((A)dS Carroll).}
 
\begin{document}
\maketitle

\section{Introduction}
\label{sec:introduction}

Carrollian~\cite{MR0192900,SenGupta1966OnAA} and
fractonic~\cite{Chamon:2004lew,Haah:2011drr,Vijay:2015mka,Vijay:2016phm}
theories both put mobility restrictions on their elementary particles
and play a prominent rôle in exciting recent advances in high energy
and condensed matter physics. But this unconventional feature and its
connected exotic symmetries also mean that many of the properties we
take for granted for Lorentz-invariant theories now need to be
reconsidered. One of them is the very definition of an elementary
system or particle with a specific symmetry.

In this work we define, classify, and analyse classical carrollions
(i.e., Carroll particles) and fractons, which for the purposes of this
paper we take to mean particles with conserved electric charge and
dipole moment. We propose that fractons are 
elementary systems with fracton symmetry; in other words, homogeneous
symplectic manifolds of the dipole group~\cite{Gromov:2018nbv} which,
since this group has vanishing symplectic cohomology, are nothing but
its coadjoint orbits. This provides, at least at a classical level,
one answer to the question \emph{``What are fractons?''}. Indeed, we
show that our proposed fractons indeed satisfy the expected
properties. This definition follows Souriau~\cite{MR1461545}, which
may be interpreted as a classical analogue of Wigner's classification
of Poincaré particles~\cite{Wigner:1939cj}.

Most of our discussion applies equally to both carrollions and
fractons. The close relation between their underlying
symmetries~\cite{Bidussi:2021nmp} (see also~\cite{Marsot:2022imf})
reveals an interesting correspondence between them. More precisely, they both
possess conserved angular momentum $\jb$ and linear momentum $\p$, but
Carroll energy $E$ and center-of-mass charge $\k$ can be reinterpreted
as fracton electric charge $q$ and dipole moment $\db$, respectively.
At the heart of our discussion are the following commutation relations
\begin{align}
 [\k,\p] &= E\quad  (\text{Carroll})  &\Leftrightarrow&&   [\db,\p] &= q \quad  (\text{Fracton})
\end{align}
from which many of the unusual properties emerge.  We should remark
that the symmetries do not precisely match: fracton energy has no counterpart in the
Carroll world, but since it is central in the dipole algebra
it does not affect the classification of coadjoint orbits and hence
may be safely ignored in most of our discussion.

Classical Carroll/fracton particles (i.e., the coadjoint orbits of the
Carroll/dipole group) fall broadly into two classes,
depending on whether or not the carrollian energy (dually, the fracton
charge) vanishes.  If nonzero, we call these particles \emph{massive
  carrollions} for reasons we will explain in the bulk of the paper
and, similarly, if the carrollian energy is zero, we call the
corresponding particles \emph{massless}. Carroll/fracton correspondence
relates massive carrollions to charged monopoles, both of which share
the characteristic feature of being stuck to a point: (see
Figure~\ref{fig:fracton_mobmom_lim})
\begin{align}
  \dot \x = \zerob.
\end{align}
Whereas for carrollions this feature is due to the conservation of
center-of-mass $\k=E\x$, for monopoles it is due to the conservation
of the dipole moment $\db=q\x$. This conclusion can be systematically derived from a
phase space action associated with the relevant coadjoint orbit, which
for the massive spinless carrollion is given by
\begin{equation}
  S_{\text{massive}}=\int d\tau 
  \left[
-E\dot{t}+\boldsymbol{\pi}\cdot\dot{\boldsymbol{x}}-N\left(E-E_{0}\right)
  \right]
\end{equation}
and for the spinless fracton monopole by
\begin{equation}
  S_{\text{monopole}}=\int d\tau 
  \left[
-E\dot{t}+\boldsymbol{\pi}\cdot\dot{\boldsymbol{x}}+\Phi\dot{q}-N\left(E-E_{0}\right)-\eta\left(q-e\right) 
  \right] \, .
\end{equation}
In these particle actions, all quantities are varied except for the
fixed values $E_0$ for the energy and $e$ for the charge.  We also
see that once we fix the time $\tau=t$ and solve the constraints we
arrive at equivalent actions on the reduced phase space
\begin{align}
  S_{\text{red}} = \int dt 
  \left[
 \boldsymbol{\pi}\cdot\dot{\boldsymbol{x}}-E_{0} 
  \right] \, .
\end{align}
In this sense they are intrinsically the same.

\begin{figure}[h!]
\centering
  \begin{tikzpicture}[scale=0.4]
          \draw[] (0,0) circle (2cm and 0.5cm); 
          \draw[] (2,0) -- (2,8) arc (360:180:2cm and 0.5cm) -- (-2,0) arc (180:360:2cm and 0.5cm); 
          \draw[] (0,8) circle (2cm and 0.5cm);

            \draw[color=green,thick] plot [smooth ,tension=2] coordinates {(-1,0) (-1,8)};
            \filldraw[color=green, fill=blue!5, very thick](-1,0) circle (0.2);
            \filldraw[color=green, fill=blue!5, very thick](-1,4) circle (0.2);
            \filldraw[color=green, fill=blue!5, very thick](-1,8) circle (0.2);

             \draw[color=gray!60,thick] plot [smooth ,tension=1] coordinates {(1,0) (1.3,2) (1,4) (0.6,6) (1,8)};
            
            \draw[color=gray!60,->,thick] (1,0) -- (1+0.7,0-0.5);
            \draw[color=gray!60,->,thick] (1,4) -- (1+0.7,4-0.5);
            \draw[color=gray!60,->,thick] (1,8) -- (1+0.7,8-0.5);

            \node at (1.7,4.4) {{\color{gray}$\db$ }};
            \draw[-Stealth] (4,0) -- (4,8);
            \node[right] at (4,8) {$t$};

          \end{tikzpicture}
          \caption{This figure is a sketch of the movement of fractons
            (and carrollions, dually) in space. When the dipole moment
            is conserved, monopoles are restricted to a point in space
            as pictured by the straight left line.  On the other hand,
            the mobility of elementary dipoles is not restricted:
            dipole conservation $\dot \db =\zerob$ implies the
            dipole vector is inert (right).  Dually, massive
            carrollions are stuck at a point since the center-of-mass
            charge is conserved as can be visualised by thinking of a
            massive Poincaré particle for which the light cone closes
            to a line. For the massless carrollian centrons again
            mobility is not restricted.}
    \label{fig:fracton_mobmom_lim}
\end{figure}
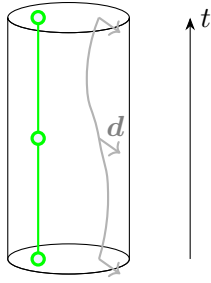

On the other hand, certain massless carrollions (which we tentatively
call \emph{centrons}) having vanishing energy and momentum but
nonzero center-of-mass charge correspond to elementary fractonic
dipoles.  Their mobility, as expected, is not restricted as we can see
from their phase space action
\begin{equation}
  \label{eq:S_dipolewithx}
  S_{\mathrm{dipole}}= \int dt 
  \left[
    \db\cdot\dot{\vb} + \bm{\pi} \cdot \dot \x -\eta\left(\|\db\|^{2}-d^{2}\right)  -\bm{u} \cdot \bm{\pi}
  \right] \, ,
\end{equation}
where only $d$ is not varied (the centron action is given by
$\db \mapsto \k$). In this action $\xb$ represents the position of the
dipole and $\db$ is a vector that represents the dipole moment,
anchored at $\xb$. They are independent degrees of freedom in the
action. Variation with respect to $\bm{\pi}$ shows that
$\dot \x = \bm{u}$, but since $\bm{u}$ is an arbitrary Lagrange
multiplier, $\xb$ is unconstrained, i.e., without coupling to external
sources their time evolution is undetermined and not fixed by the
symmetries.  On the other hand varying with respect to $\vb$ leads to
$\dot \db = \zerob$, which shows that the dipole moment is indeed
conserved. The constraint $\|\db\|^{2} = d^{2}$ fixes the norm of the
dipole momentum, where $d$ is an external parameter characterising the
magnitude of the dipole moment. This leaves us with two remaining
degrees of freedom, which are the two angles that determine the
direction of the dipole moment vector with norm $d$.
  
We may also add spin for the massive carrollion/monopole and in this
way also uncover various other massless carrollions/neutral fractons,
which have not appeared in the literature, and show that the zoo of
carrollions and fractons is more diverse than discussed so far (see
Table~\ref{tab:carrollions}). The good news is that notwithstanding
these particles being physically distinct, they are related by
$GL(2,\RR)$ transformations arising as outer automorphisms of the
Carroll/dipole group.  This implies that at a technical level we may
often reuse results for one type of particle and apply them to another
type, simply by acting with a $GL(2,\RR)$ transformation.

We also comment on the connection to Poincaré particles (see in
particular Figure~\ref{fig:mom_lim}), the relation to (electric and
magnetic) carrollian field theories and provide some remarks on
Carroll boost versus dipole symmetry for field theories. We also
highlight that there is a curved generalisation to (A)dS Carroll and
fractons on curved space.

The rest of the paper is organised as follows.

We start in Section~\ref{sec:carr-part-3+1} by briefly reviewing
Carroll symmetries and describing the classification of carrollions,
which is summarised in Table~\ref{tab:carrollions}. A more
comprehensive discussion of Carroll/dipole symmetry (in arbitrary
dimension) can be found in Appendix~\ref{app:carroll-symmetry},
whereas Appendix~\ref{sec:coadjoint-orbits-n=3} contains more details
about the classification of the coadjoint orbits (in three spatial
dimensions), including their structure as manifolds.

In Section~\ref{sec:particle-actions} we provide a systematic
derivation of particle actions for carrollions.

In Section~\ref{sec:carr-vs-massl} we discuss the limit from Poincaré
particles to carrollions and highlight that for to each spinning
massive Poincaré particle there corresponds a massive spinning
carrollion. The massless carrollions mostly derive from Poincaré
tachyons and the massless Poincaré particles seem to vanish in the
limit.

In Section~\ref{sec:carr-vers-fract} we introduce dipole
symmetries and the Carroll/fracton correspondence. We then discuss the
mobility restrictions of the monopoles and dipoles. As an instructive
example we also show how the elementary dipoles emerge from two
monopoles.

In Section~\ref{sec:field-theor-gener} we comment on the
generalisation to field theories and curved space. In
Section~\ref{sec:electr-magn-carr} we discuss the relation of massive
and massless carrollions to Carroll field theories. Massive Carroll
particles are related to what is sometimes called ``electric'' theory.
Massless carrollions with helicity seem to be related to ``magnetic''
theories, but the known actions have an additional source term, for
which we propose a possible resolution. In
Section~\ref{sec:dipole-vers-carr} we contrast Carroll boost symmetry
with dipole symmetries. In Section~\ref{sec:ads-carroll-fractons} we
describe how the Carroll/fracton correspondence can be generalised to curved
space, more precisely, (A)dS Carroll.

In Section~\ref{sec:discussion-outlook} we comment on various
potential applications of our results such as to time-like symmetries,
other exotic particles, flat holography and black holes.

\section{Classical Carroll particles: Coadjoint orbits}
\label{sec:carr-part-3+1}

In this Section we summarise the classification of classical
carrollions (i.e., Carroll particles) in $3+1$ spacetime dimensions.
In other words we classify coadjoint orbits of the ($3+1$)-dimensional
Carroll group. This summary should be read together with
Table~\ref{tab:carrollions}.  Readers interested in the
mathematical details are encouraged to read the two appendices.  In
Appendix~\ref{app:carroll-symmetry}, whose point of departure is the
Carroll Lie algebra in general dimension, we show how to view the
Carroll group as a matrix group, allowing us to determine the adjoint
and coadjoint actions explicitly.  In
Appendix~\ref{sec:coadjoint-orbits-n=3} we go through the systematic
classification  of coadjoint orbits as well as providing more details
about their structure.

We may give Carroll spacetime global coordinates $(t,\x)$ on which the
Carroll group $G$~\cite{MR0192900,SenGupta1966OnAA} acts via the
following three kinds of Carroll transformations
\begin{itemize}
\item \textbf{rotations}: $(t,\x) \mapsto (t, R\x)$, where $R \in
  \SO(3)$;
\item \textbf{(carrollian) boosts}: $(t,\x)\mapsto (t +
  \bv\cdot\x,\x)$;
\item and \textbf{translations}: $(t,\x) \mapsto (t + s, \x + \ba)$.
\end{itemize}
The distinguishing feature of Carroll symmetries are the carrollian
boosts which only act on time, rather than time and space, like
Poincaré, or only space, like for Galilei. The general Carroll
transformation $(R,\bv,\ba,s)$ is a composition of rotations, boosts
and translations and is given by (see Appendix~\ref{sec:symm-brok-param}):
\begin{equation}
  \label{eq:carroll-trans}
  (t,\x) \mapsto (t + s + \bv \cdot \x , R\x + \ba).
\end{equation}

In this section we will restrict to the ($3+1$)-dimensional Carroll
group, whose Lie algebra $\g$ is spanned by $J_i, B_i, P_i, H$ with
nonzero Lie brackets
\begin{align}
  \label{eq:3-carroll-algebra}
  [J_i, J_j] &= \epsilon_{ijk} J_k &
  [J_i, B_j] &= \epsilon_{ijk} B_k &
  [J_i, P_j] &= \epsilon_{ijk} P_k &
  [B_i,P_j] &=\delta_{ij} H,
\end{align}
where the Levi-Civita symbol $\epsilon_{ijk}$ is normalised so that
$\epsilon_{123}=1$. This defines the adjoint action of $\g$ on itself
and by exponentiation also the adjoint representation which for a
group element $g$ acts as $\Ad_{g}A = gAg^{-1}$ on $A \in \g$.
From that we define the coadjoint representation: if $\alpha \in \g^*$
is a element in the dual of the Lie algebra, then
\begin{align}
  \langle \Ad^{*}_{g}\alpha, A\rangle =   \langle \alpha, \Ad_{g^{-1}}A\rangle
\end{align}
Covectors $\alpha \in \g^*$ may be interpreted as the conserved
quantities of the elementary systems: $\alpha = (\bj,\k,\p,E)$, where
$\bj=\langle{\alpha,\boldsymbol{J}}\rangle$ is the angular momentum,
$\k = \langle{\alpha,\boldsymbol{B}}\rangle$ the centre of mass,
$\p = \langle{\alpha,\boldsymbol{P}}\rangle$ the (linear) momentum and
$E= \langle\alpha, H\rangle$ the energy. The coadjoint action of the
Carroll group element $g = (R,\bv,\ba,s)$ on $\alpha = (\bj,\k,\p,E)$
is given by $\Ad^*_g \alpha = (\bj', \k', \p', E')$
where~\cite{Duval:2014uoa}
\begin{equation}
  \label{eq:coadjoint-rep-n=3-summary}
  \begin{split}
    \bj' &= R \bj + \bv \times R\k + \ab \times R\p + E\bv \times \ab\\
    \k' &= R\k + E \ab\\
    \p' &= R\p - E\bv\\
    E' &= E.
  \end{split}
\end{equation}
The intuition is that the coadjoint action characterises the
transformation of the conserved quantities under the action of the
Carroll group. Each coadjoint orbit defines a specific
``particle''.  Because coadjoint orbits are elementary systems which
often have the interpretation of a particle, we use the two
synonymously in this work.  Of course, not every coadjoint orbit
admits such an interpretation, e.g., the vacuum.

There are two obvious Casimirs of the Carroll group \cite{MR0192900}:
$H$ and $W^2$, which define a linear and a quartic function on $\g^*$,
respectively, evaluating on $\alpha = (\bj,\k,\p,E)$ to
\begin{equation}
  \label{eq:carr-inv}
  H(\alpha) = E \qquad\text{and}\qquad W^2(\alpha) = \|E \bj + \p
  \times \k\|^2.
\end{equation}
It is clear that $E$ is invariant under the coadjoint representation
and it is a short calculation using
equation~\eqref{eq:coadjoint-rep-n=3-summary} to see that $E\bj + \p
\times \k$ transforms by a rotation, so its norm is invariant.
Therefore these functions are constant on coadjoint orbits and
therefore they are useful in their classification.

The energy $E$ is analogous to the mass of Poincaré coadjoint orbits
and we will therefore tentatively follow the convention to call the
orbits \emph{massive} when $E \neq 0$ and \emph{massless} when $E=0$.

\subsection{Massive carrollions ($E \neq 0$)}
\label{sec:e-neq-0}

As shown in Appendix~\ref{sec:coadj-orbits-with-nonzero-energy}, any
$\alpha \in \g^*$ with $E \neq 0$ can be boosted to the ``rest frame''
so that $\p = \bzero$, in analogy with the massive Galilei and
Poincaré particles. In addition, we may translate so that the centre
of mass $\k=\bzero$. Doing so brings $\alpha$ to the form
$(\spb,\bzero,\bzero,E)$, where we have introduced the intrinsic spin
vector
\begin{align}
  \label{eq:spin-def}
 \spb := \bj +  E^{-1} \p \times \k,
\end{align}
whose norm $S = \|\spb\|$ is the intrinsic spin of the particle. We
can use this quantity to differentiate between two kinds of massive
particles: \textbf{spinless massive carrollions} where $S=0$ and
\textbf{massive carrollions with spin $S>0$}.

Notice that we have nonzero angular momentum, even though $\p$
vanishes, so it is justified to call $S$ the spin.   This is in
contrast to the orbital angular momentum that even spinless
carrollions have: $\bj = E^{-1} \k \times \p$.

In energy-momentum space, the condition $E=E_0\neq 0$ fixes a
three-dimensional affine hyperplane (see Figure~\ref{fig:mom_lim})
over which the coadjoint orbit fibres. The coadjoint orbits for
spinless massive particles are then given by the cotangent bundle of
these affine hyperplanes. For the massive particles with nonzero spin,
the coadjoint orbits acquire two extra dimensions: to every point in
the phase space of the spinless particle there is associated a
$2$-sphere of radius the intrinsic spin. This is explained in
Appendix~\ref{sec:struct-coadj-orbits}. 

In summary, massive coadjoint orbits are in one-to-correspondence with
pairs $(E,S)$, where the energy $E \neq 0$ and the spin $S \geq 0$.

\subsection{Massless orbits ($E = 0$)}
\label{sec:e-eq-0}

When $E=0$ the coadjoint action reduces to
\begin{equation}
  \label{eq:coadjoint-rep-n=3-Ezero}
  \begin{split}
    \bj' &= R \bj + \bv \times R\k + \ab \times R\p\\
    \k' &= R\k\\
    \p' &= R\p \\
  \end{split}
\end{equation}
and the Casimir $W^2$ is still an invariant and on $\alpha =(\bj,
\k,\p,0)$ it now takes the value
\begin{align}
  W^2(\alpha) = \| \p \times \k\|^2 = \|\p\|^2 \|\k\|^2  - (\p \cdot \k)^2,
\end{align}
where we have used a standard vector identity. When $E=0$, it is clear
from equation~\eqref{eq:coadjoint-rep-n=3-Ezero} that $\|\p\|^2$,
$\|\k\|^2$ and $\p \cdot \k$ are separately
invariant~\cite{Duval:2014uoa}, which we may use to further refine the
classification of massless orbits. Before doing so, however, we may
already highlight two interesting physical consequences:
\begin{itemize}
\item Since $\|\p\|^2$ is invariant, massless orbits cannot be boosted
  to a ``rest frame'' where $\p = \bzero$. This similarity to the
  massless Poincaré particles is a further justification to call these
  orbits massless (this is however also a property of tachyons).
\item Another peculiar feature of the massless coadjoint orbits is the
  existence of the invariant $\|\k\|^2$. This differs from the
  Poincaré and Galilei (more precisely, Bargmann) case where we can
  always translate massive and massless particles in such a way that
  the centre of mass $\k$ vanishes.
\end{itemize}

We now return to the classification of massless orbits.  We will first
focus on the orbits where $\p \times \k =\bzero$, so that the linear
momentum $\p$ and the centre of mass $\k$ are parallel.  

\subsubsection{\texorpdfstring{Vacuum sector ($\p=\k=\bzero$)}{Vacuum sector (p=k=0)}}
\label{sec:p=k=0:-vacuum-sector}

The vacuum sector is given by restricting to $E=0$ and by additionally
setting $\p=\k=\bzero$. The remaining nontrivial coadjoint action is a
rotation of the angular momentum $\bj' = R \bj$. This gives rise to
the quadratic invariant $\| \bj \|^2$. The orbit with $\bj=\bzero$
consists of the origin in $\g^*$ and we call it the \textbf{vacuum}.
The orbits with $\bj \neq \bzero$ are $2$-spheres of radius $\|\bj\|$
and we call the corresponding particles \textbf{spinning vacua}.

In summary, these orbits are parametrised by $j = \|\bj\| \geq 0$ and
they may be uniquely specified by the equations $E=0$,
$\p=\bzero$, $\k = \bzero$ and $\|\bj\| = j \geq 0$.

\subsubsection{Massless (parallel) carrollions}
\label{sec:massless-parallel}

These are the orbits where $\p \times \k = \bzero$ but not both $\p$
and $\k$ are zero.  These orbits are all related by the action of
automorphisms, as explained in Appendix~\ref{sec:orbits-mod-autos}.
The action of automorphisms on momenta is described by
equation~\eqref{eq:autos-momenta}, where we see that automorphisms act
on $\p,\k$ via general linear transformations which act transitively
on lines.  Nevertheless, the orbits are described by different
equations and have different physical interpretations, so it is worth
looking at them separately.

Let us first set $\k = \bzero$, but $\p \neq \bzero$.  The coadjoint
action reduces to
\begin{equation}
  \label{eq:coadjoint-rep-n=3-kzero}
  \begin{split}
    \bj' &= R \bj +  \ab \times R\p\\
    \p' &= R\p \, .
  \end{split}
\end{equation}
We see that not just is $\|\p\|$ invariant, but also $\bj \cdot
\p$.  We define the \emph{helicity} $h$ by
\begin{align}
  \label{eq:helicity-def}
  \bj \cdot \p = h\|\p\| \, .
\end{align}
The helicity can be any real number. As shown in
Appendix~\ref{sec:p-neq-zero}, we may always bring such a covector to
the form $(\bj,\bzero,\p,0)$, where $\bj = h \p/\|\p\|$, where $\p$
can be rotated into any desired direction. For example, we can take
$\p = (0,0,p)$ and hence $\bj = (0,0,h)$. While we will call these
particles~\textbf{massless carrollions with helicity $h$} it would be
equally justified to call them \textbf{aristotelions} which emerge
from the (flat) aristotelian space (see,
e.g.,~\cite{Figueroa-OFarrill:2018ilb}) and which have no boost
symmetries and thus no center-of-mass conservation and precisely the
coadjoint action~\eqref{eq:coadjoint-rep-n=3-kzero} with $E'=E$.

These orbits are cut out by the equations $E=0$, $\k = \bzero$,
$\|\p\|=p\neq 0$ and $\bj \cdot \p = h p$, depending on the two
parameters $p>0$ and $h \in \RR$.

Everything we just discussed applies mutatis mutandis to the case
where $\p=\bzero$ and $\k \neq \bzero$. We may always bring a covector
in such an orbit to the form $(\bj, \k, \bzero,0)$, where
$\bj = h \k/\|\k\|$ and $h \in \RR$. Such orbits are characterised by
the equations $E=0$, $\p = \bzero$, $\|\k\|=k \neq 0$ and
$\bj \cdot \k = h k$. They characterise the center of mass and we
therefore tentatively call them \textbf{centrons}.

Finally, we discuss the case of $\p \times \k = \bzero$, but neither
$\p$ nor $\k$ are zero.  This breaks up into two cases depending on
whether $\p$ and $\k$ are parallel or antiparallel.  In this case,
since $E=0$, the inner products $\bj \cdot \p$ and $\bj \cdot \k$ are
constant on the orbit.  Also $\p \cdot \k$ is constant, but since
$\p \times \k = \bzero$, it follows that
$\p \cdot \k = \pm \|\p\|\|\k\|$, where the plus sign says the angle
between them is $0$ (parallel) and the minus sign says the angle
between them is $\pi$ (antiparallel).  It is also the case that
$\bj\cdot \p$ and $\bj \cdot \k$ are not independent, so that if we
know $\|\p\|$, $\|\k\|$, $\p \cdot \k$ and one of $\bj\cdot\p$ or
$\bj \cdot \k$, we know them all.  We may impose these conditions in
order.  We start with the 10-dimensional $\g^*$ and impose $E=0$ to
drop down to a $9$-dimensional hyperplane with coordinates
$(\bj,\k,\p)$.  We impose $\|\p\|=p>0$ and $\|\k\|=k>0$ and we get a
$7$-dimensional manifold diffeomorphic to
$\RR^3 \times S^2 \times S^2$.  Next we impose
$\p \times \k = \bzero$.  This results in two disconnected
$5$-dimensional submanifolds $\RR^3 \times M^2_\pm$ one for each sign
in $\p\cdot\k = \pm p k$. The two-dimensional manifolds $M^2_\pm$ are
submanifolds of $S^2 \times S^2$ and consist of points
$(\pm k \bu, p \bu)$, where $\bu$ is a unit vector in the direction of
$\p$.  Notice that both $M^2_\pm$ are diffeomorphic to $S^2$, so that
each of the five-dimensional manifolds is diffeomorphic to
$\RR^3 \times S^2$.  Finally, we impose $\bj \cdot \p = h p$, which
cuts the dimension by one, resulting in a $4$-dimensional orbit.  In
summary, the condition $\p \times \k = \bzero$ does not impose three
relations as one might naively suspect, but only two sets of two
relations, distinguished by a sign.

\subsubsection{Generic massless carrollions}
\label{sec:generic}

It remains to discuss the orbits where $E=0$ but
$\p \times \k \neq \bzero$, which is the generic case and motivated us
to call them~\textbf{generic massless carrollions}. This implies that
$\|\p\|=p>0$ and $\|\k\|=k>0$ and hence $\p \cdot \k = p k \cos\theta$
for some angle $\theta \in (0,\pi)$. Let us write $\p = p \bu$, for
$\bu$ a unit vector in the direction of $\p$. Then
$\k - k\cos\theta\bu$ is perpendicular to $\p$ and hence to $\bu$. Let
us write it as $\k - k\cos\theta\bu = k \sin\theta \bu_{\perp}$, where
$\bu_{\perp}$ is a second unit vector that is perpendicular to $\bu$.
Therefore we may write
$\k = k (\cos\theta\bu + \sin\theta\bu_{\perp})$.

Since $\p \times \k \neq \bzero$ we may use
boosts and translations to set $\bj = \bzero$. This should be
contrasted with the orbits of Sections~\ref{sec:p=k=0:-vacuum-sector}
and \ref{sec:massless-parallel} where this is not possible when
$h\neq 0$. Related to this observation is the absence of a notion of
spin for the generic massless orbits.

As expected for $\theta = 0$, where $\p$ and $\k$ are parallel and
$\theta = \pi$, where they are antiparallel, we are back to the
earlier case of Section~\ref{sec:massless-parallel}, except that the
orbit drops dimension from $6$ to $4$.

\begin{table}
  \centering
  \caption{Overview of carrollions}
  \setlength{\extrarowheight}{3pt}
  \resizebox{\linewidth}{!}{
    \begin{tabular}{>{$}l<{$}l>{$}l<{$}>{$}c<{$}>{$}l<{$}}
      \toprule
      \multicolumn{1}{l}{\#} & \multicolumn{1}{c}{Particle description} & \multicolumn{1}{c}{Orbit representative} & \dim\mathcal{O}_\alpha & \multicolumn{1}{c}{Equations for orbits}\\
                             & & \multicolumn{1}{c}{$\alpha=(\bj, \k, \p,E)$} & & \\ \midrule \rowcolor{blue!7}
      1  & Massive spinless & (\bzero, \bzero, \bzero,E_0) & 6 & E =E_0 \neq 0, E_0 \bj + \p \times \k = \bzero\\ 
      2  & Massive with spin $S$ & (S\bu, \bzero,  \bzero , E_0)  & 8 & E =E_0\neq 0,\|\bj+E_0^{-1}\p \times\k\| = S >0 \\ \midrule \rowcolor{blue!7}
      3  & Vacuum & (\bzero,\bzero,\bzero,0)  & 0  & E=0,\p=\bzero,\k=\bzero,\bj=\bzero\\     
      4  & Spinning vacuum  & (j\bu,\bzero,\bzero,0) & 2 & E=0,\p=\bzero,\k=\bzero,\|\bj\|=j>0 \\ \cmidrule{1-2} \rowcolor{blue!7}
      5  & Centrons & ( h \bu, k \bu, \bzero, 0) & 4 & E=0,\p=\bzero,\|\k\|=k>0, \bj\cdot\k = h\|\k\| \in \RR \\
      6  & Massless with helicity $h$ & (h \bu, \bzero, p \bu, 0) & 4 & E=0,\k=\bzero,\|\p\|=p>0, \bj \cdot \p = h \|\p\|\in \RR \\     \rowcolor{blue!7}
      7_+ & Parallel massless & (h\bu, k\bu, p\bu, 0) & 4 & E=0,\|\p\|=p>0, \|\k\|=k>0, \p\cdot\k = p k, \bj \cdot \p = h \|\p\|\in \RR \\
      7_- & Antiparallel massless & (h\bu, -k\bu, p\bu, 0) & 4 & E=0,\|\p\|=p>0, \|\k\|=k>0, \p\cdot\k = -p k, \bj \cdot \p = h \|\p\|\in \RR \\
      \rowcolor{blue!7}      
      8  & Generic massless &  (\bzero, k \cos\theta \bu + k\sin\theta \bu_\perp, p \bu,0)  & 6 & E=0,\|\p\|=p>0, \|\k\|=k>0, \p\cdot\k = pk\cos\theta, \theta \in  (0,\pi) \\
      \bottomrule
    \end{tabular}
  }
  \caption*{This table provides an overview of the carrollions (=
    Carroll particles) which are summarised in
    Section~\ref{sec:carr-part-3+1}. As indicated by the horizontal
    line they are roughly separated into massive ($E\neq 0$) and
    massless ($E=0$) orbits. The shorter horizontal line separates the
    vacuum sector from the massless particles. The second column
    provides a tentative descriptive name (if one exists). The third
    column displays an orbit representative: the notation is such that
    $\bu \in \RR^3$ represents a fixed unit-norm vector and in the
    last row $\bu_{\perp}\in \RR^3$ is a second unit-norm vector
    perpendicular to $\bu$. The last column provides the equations
    which define the orbits. One can easily check that
    $10 - \#\text{equations} = \dim\mathcal{O}_{\alpha}$ in all cases
    but the (anti)parallel massless, which might seem to be
    under-constrained but as discussed in the text they are not.}
  \label{tab:carrollions}
\end{table}

\section{Particle actions}
\label{sec:particle-actions}

Given a coadjoint orbit there is a systematic to way to associate with
it a particle action. This is fundamental to the symplectic approach
to dynamical systems pioneered by Souriau \cite{MR1461545}. Since the
actions provide information concerning the mobility of the particle
and are useful for many applications, e.g., for path integral
quantisation~\cite{Alekseev:1988vx}, we will provide them in this
section and analyse their (classical) properties. For each carrollion
we follow the method explained below and
Appendix~\ref{sec:maurer-cartan-one}, but see also
\cite[§4.4.5]{Bergshoeff:2022eog}, \cite[§§2,3]{Barnich:2022bni} and
\cite[§5]{Oblak:2016eij} for further useful details and references.

Let $\tau \mapsto g(\tau)$ be a curve in the Carroll group $G$.
The action corresponding to the coadjoint orbit of $\alpha \in \g^*$
is given by
\begin{align}
  \label{eq:Lcoadj}
  S=\int L d\tau  =\int \langle \alpha,g^{-1} \dot g \rangle d\tau \, .
\end{align}
The term $g^{-1}\dot g$, where $\dot g = \frac{dg}{d\tau}$, is the
pull-back of the left-invariant Maurer--Cartan
form~\eqref{eq:LI-MC-one-form}. For the case at hand the lagrangian is
then given by
\begin{align}
  \label{eq:action=general}
  L[R(\varphib),\bv,\x,t]  
  &= \tfrac12 \Tr \left( J^T R^T \dot R  \right)
  + (R\k)\cdot \dot\bv
  + (R\p)\cdot  \dot\x
  + E 
    \left(
    \dot t
  +  \tfrac12  \x \cdot \dot\bv
    -\tfrac12  \bv \cdot \dot\x
    \right) \\
  &= \tfrac12 \Tr 
    \left(
    J^T R^T \dot R
    \right)
  + (R\k + \tfrac12 E \x) \cdot \dot\bv
  + (R\p -\tfrac12 E \bv)\cdot \dot\x
  + E \dot t
\end{align}
The terms in square brackets denote the quantities that are varied (in
this case, $\varphib$, $\vb$, $\boldsymbol{x}$, $t$) while the remaining
quantities are not, i.e., $J=\varepsilon(\bj)$, $\k$, $\p$, $E$ are fixed.

The action has a global Carroll symmetry as can be seen from
$g^{-1} \dot g$ which is invariant under $\tau$-independent left
multiplication $g \mapsto hg$.  Infinitesimally these transformations
can be parametrised by
$\left(\boldsymbol{\lambda},\boldsymbol{\beta},\boldsymbol{a},s\right)$
and act as
\begin{align}
  \label{eq:Carroll_transf}
  \delta R&=\varepsilon(\bm{\lambda})R &
  \delta \vb&=\boldsymbol{\lambda} \times \vb+\boldsymbol{\beta} &
  \delta \boldsymbol{x}&=\boldsymbol{\lambda} \times \boldsymbol{x}+\boldsymbol{a} &
  \delta t&=s+\tfrac{1}{2}(\boldsymbol{\beta}\cdot\boldsymbol{x}-\boldsymbol{a}\cdot\vb) \, ,
\end{align}
where $\varepsilon(\bm{\lambda})_{ab} = -\epsilon_{abc}\lambda^{c}$
(see Appendix~\ref{sec:case-n=3}). This leads to the following Noether
charges, which we denote with a subscript $Q$:
\begin{equation}
  \label{eq:Q-general}
  \begin{split}
    \bj_{Q}&= R \bj + \bv \times R\k + \x \times R\p + E\bv \times  \x \\
    \k_{Q} &= R\k + E \x \\
    \p_{Q} &= R \pb -E\vb \\
    E_{Q} &= E.
  \end{split}
\end{equation}
This shows that the Noether charges are given by a coadjoint action on
$\alpha$.

These actions have gauge symmetry given by a right action
$g \mapsto g h(\tau)$, where $h(\tau)$ has to be in the stabiliser of
$\alpha$.  Since the stabiliser, and consequently the physical degrees
of freedom and constraints, depends on the specific particle, we will
now analyse them case by case.  As shown in Appendix
\ref{sec:maurer-cartan-one} (but see, e.g.,
also~\cite{Barnich:2022bni}), particle actions only depend on the
coadjoint orbit of $\alpha$ and not on $\alpha$ itself.  Therefore we
will feel free to choose a convenient representative in the following
in order to simplify our computations.  We will also neglect
subtleties related to quantisation, like, e.g., boundary terms.

\subsection{Massive spinless carrollion action}
\label{sec:action-1}

Massive spinless carrollions can always be brought into the ``rest
frame''
\begin{align}
  \label{eq:massive-s0-rep}
  \alpha=\left(0,\bzero,\bzero,-E_{0}\right),
\end{align}
where $E_{0}$ is a constant. The lagrangian~\eqref{eq:action=general}
is then up to a total derivative of the form
\begin{equation}
  \label{eq:L_E_Lag}
  L[\bv,\x,t]  =E_{0}\left(\vb\cdot\dot{\boldsymbol{x}}-\dot{t}\right) \, .
\end{equation}
We have used the fact that the rotations are part of the stabiliser to
go to a reduced phase space without $\varphib$. As discussed, this
action is by construction invariant under Carroll symmetries and the
Noether charges are given by restricting~\eqref{eq:Q-general} to
our choice of representative~\eqref{eq:massive-s0-rep}.

To express the lagrangian in Hamiltonian form, it is necessary to
introduce the canonical momenta
\begin{align}
  \boldsymbol{\pi}:=\frac{\partial L}{\partial\dot{\boldsymbol{x}}}=E_{0}\vb\,,\qquad E:=-\frac{\partial L}{\partial\dot{t}}=E_{0}.
\end{align}
The field $\vb$ is associated with the canonical momentum
$\boldsymbol{\pi}$ conjugate to $\boldsymbol{x}$.  In other words,
the first term in~\eqref{eq:L_E_Lag} is already in Hamiltonian form
when $\bm{\pi}=E_{0}\bm{b}$. The lagrangian of the massive Carroll
particle in a canonical form can then be written as
\begin{equation}
  L_{\text{can}}[\x,t,\bm{\pi},E,N] = - E\dot{t} + \boldsymbol{\pi}
  \cdot \dot{\boldsymbol{x}}-N\left(E-E_{0}\right),\label{eq:Lag_E_Hamiltonian}
\end{equation}
where $N$ is the Lagrange multiplier that enforces the constraint
$E-E_{0} = 0$. This restricts to a unique orbit and is the very
same constraint that we also found in Section~\ref{sec:e-neq-0}, cf.,
the planes in Figure~\ref{fig:mom_lim}.

A related action  was previously found
in~\cite{Bergshoeff:2014jla}, where the lagrangian is invariant under
time reversal, i.e., the constraint was chosen to be of the form
$E^{2}-E_{0}^{2}= 0$, describing both carrollian particles and
anti-particles simultaneously, i.e., it does not restrict to just one
orbit.  On the other hand, the lagrangian~\eqref{eq:Lag_E_Hamiltonian}
describes a particle, or an antiparticle, depending on the sign of
$E_{0}$. This case was also discussed
in~\cite[Appendix~A]{Duval:2014uoa} and \cite{deBoer:2021jej}.

Solving the constraint $E-E_{0} = 0$, and imposing the gauge
condition $t=\tau$, one finds the lagrangian describing the dynamics
in the reduced phase space
\begin{align}
  \label{eq:action-massive-s0}
  L_{\text{red}}\left[\boldsymbol{x},\boldsymbol{\pi}\right]=\boldsymbol{\pi}\cdot\dot{\boldsymbol{x}}-E_{0} \, .
\end{align}
It depends on 6 independent canonical variables, which is precisely
the dimension of the coadjoint orbit (see
Table~\ref{tab:carrollions}, orbit $\#1$), as it must.

The just discussed actions imply that
\begin{align}
  \dot \x &= \zerob & \dot{\bm{\pi}} &= \zerob & E&=E_{0} \, ,
\end{align}
i.e., an isolated massive Carroll particle does not move, has constant
conjugate momentum $\bm{\pi}$ and fixed energy $E_{0}$. Let us
emphasise that the velocity $\dot \x$, which is bound to be zero, is
not connected to the momentum $\bm{\pi}$ which, although constant, can
be nonzero. We will have more to say about these particles in
Section~\ref{sec:carr-vers-fract}, where we also highlight the
relation to fractons (more precisely, monopoles).

\subsubsection{Infinite symmetries of massive carrollions}
\label{sec:infin-symm-mass}

The action functional of the massive carrollion is invariant not only
under Carroll transformations but actually under an infinite-dimensional
symmetry. The lagrangian in \eqref{eq:Lag_E_Hamiltonian} is invariant
(up to boundary terms) under canonical transformations generated by an
arbitrary function
$F=F\left(E,\boldsymbol{x},\boldsymbol{\pi}\right)$. For a canonical
variable $z$, its transformation law is given by
$\delta z=\left\{ z,F\right\} $. Explicitly, they read
\begin{align}
\label{eq:transf_gen_symm}  
  \delta t&=-\frac{\partial F}{\partial E} &
  \delta E&=0 &
  \delta\boldsymbol{x}&=\frac{\partial F}{\partial\boldsymbol{\pi}} &
  \delta\boldsymbol{\pi}=-\frac{\partial F}{\partial\boldsymbol{x}}                                       
\end{align}
together with $\delta N=0$. 

The conserved charge associated to this symmetry is the function
$F\left(E,\boldsymbol{x},\boldsymbol{\pi}\right)$ itself. Indeed,
using the equations of motion
$\dot{E}=\dot{\boldsymbol{x}}=\dot{\boldsymbol{\pi}}=0$, one can
directly show that
\begin{align}
  \dot{F}\left(E,\boldsymbol{x},\boldsymbol{\pi}\right)=0 \, ,
\end{align}
i.e., it is conserved under $\tau$ evolution. This discussion also
generalises for the massive spinning carrollions.

There are some interesting particular cases. For example, the Carroll
transformations are obtained from the following generator
\begin{equation}
  \label{eq:F-Carroll}
  F
  = \boldsymbol{\lambda} \cdot \left(\boldsymbol{x} \times \boldsymbol{\pi}\right)
    - E \bm{\beta} \cdot \boldsymbol{x} + \boldsymbol{a} \cdot \boldsymbol{\pi} - s E \, .
\end{equation} 
In addition, the weak carrollian structure of Carroll
spacetime~\cite{Henneaux:1979vn,Duval:2014lpa} has symmetries that include
non-linear transformations, which can be obtained using the following
generator
  
\begin{equation}
  F=-Ef(\x) \, , 
\end{equation}
with $f(\x)$ being an arbitrary function of the coordinates which
generalises the linear Carroll boosts $f(\x)=\bm{\beta}\cdot \x$. Thus,
using~\eqref{eq:transf_gen_symm} one finds
\begin{align}
  \delta t&=f(\boldsymbol{x}) &
  \delta  E&=0 &
  \delta\boldsymbol{x}&=0 &
  \delta\boldsymbol{\pi}&=E\frac{\partial f(\x)}{\partial\boldsymbol{x}} \,.
\end{align}
This symmetry is also present in the action of the massive carrollian
scalar field (``electric'') which we discuss in
Section~\ref{sec:electr-magn-carr}.

\subsection{Massive spinning carrollion action}
\label{subsec:Massive-spinning-Carroll}

In order to incorporate a non-vanishing spin, we choose the following representative
of the coadjoint orbit
\begin{align}
  \label{eq:repr-massive-s}
  \alpha=\left(J,\bzero,\bzero,-E_{0}\right).
\end{align}
Here, $E_{0}$ is a constant and $J$ is a generic angular momentum
matrix (see Appendix~\ref{app:carroll-symmetry}). Both of these
quantities are considered as fixed parameters in the action.

From~\eqref{eq:action=general} we obtain the lagrangian
\begin{align}
  L[R(\varphib),\bv,\x,t]
  & =\frac{1}{2}\Tr\left(J^{T}R^{T}\dot{R}\right)+E_{0}\left(\vb\cdot\dot{\boldsymbol{x}}-\dot{t}\right).\label{eq:actionWithS1}
\end{align}
Here, $R=R\left(\varphib\right)$ is a rotation matrix that
depends on three independent angles $\varphi^{a}$. We have not yet specified
the particular dependence of $R$ in terms of these angles. Consequently,
the lagrangian can be rewritten as
\begin{align}
  \label{eq:actionWithS2}
  L[R(\varphib),\bv,\x,t]
  & =\frac{1}{2}\Tr\left(J^{T}R^{T}\frac{\partial R}{\partial\varphi^{a}}\right)\dot{\varphi}_{\;}^{a}
    +E_{0}\left(\vb\cdot\dot{\boldsymbol{x}}-\dot{t}\right).
\end{align}
By comparing this lagrangian with that of the spinless
case~\eqref{eq:L_E_Lag}, it is evident that there is an additional
term present which takes into account the spin degrees of freedom,
$\varphi^{a}$. The presence of additional degrees of freedom due to
spin is also a feature of actions describing Poincaré invariant
particles (see, e.g.,
\cite{Hanson:1974qy,Casalbuoni:1975hx,Brink:1976uf,Berezin:1976eg}
and~\cite{Hanson:1976cn} for a useful summary).

This action is again invariant under Carroll symmetry with Noether
charges given by~\eqref{eq:Q-general} restricted
to~\eqref{eq:repr-massive-s}. In particular since
\begin{align}
  \bj_{Q} - E_{Q}^{-1} \p_{Q} \times \k_{Q} = R \bj
\end{align}
we see that these particles have spin (cf.~\eqref{eq:spin-def})
\begin{align}
 \spb \cdot \bm{\lambda}  =\frac{1}{2}\Tr\left(J^{T}R^{T}\varepsilon(\bm{\lambda} )R\right) =R \bj\cdot \bm{\lambda}
\end{align}
which describes the non-orbital part of the total angular momentum
with $\|\boldsymbol{S}\|^{2}=S^{2}$ an invariant quantity, as
expected.

We shall use the following parametrisation of the rotation matrix 

\begin{equation}
R(\varphib)=e^{\varphi_{1}\varepsilon_{1}}e^{\varphi_{2}\varepsilon_{2}}e^{\varphi_{3}\varepsilon_{3}}.\label{eq:Coords_sphere}
\end{equation}
(where explicitly $(\varepsilon_{a})_{bc}=-\epsilon_{abc}$,
see~\eqref{eq:epsilon-rep}). In addition, by an appropriate rotation
we can always choose to align $J$ with the $z$-axis, i.e., $J=S\varepsilon_{3}$
for a constant $S>0$. Using this parametrisation, the
lagrangian~\eqref{eq:actionWithS2} takes the form
\begin{equation}
L\left[\varphib,\vb,t,\x\right]=S\left(\dot{\varphi}_{3}+\sin\left(\varphi_{2}\right)\dot{\varphi}_{1}\right)+E_{0}\left(\vb\cdot\dot{\boldsymbol{x}}-\dot{t}\right),\label{eq:Action_withS_frame}
\end{equation}
and the spin vector becomes
\[
\boldsymbol{S}=S\hat{\boldsymbol{n}},
\]
where $\hat{\boldsymbol{n}}$ is a unit vector given by
\begin{equation}
\hat{\boldsymbol{n}}=\left(\sin\varphi_{2},-\sin\varphi_{1}\cos\varphi_{2},-\cos\varphi_{1}\cos\varphi_{2}\right).\label{eq:normal}
\end{equation}
Therefore, $\|\boldsymbol{S}\|^{2}=S^{2}$ is an invariant
quantity, as expected.

The canonical momenta associated with the spin degrees of freedom
are given by

\begin{equation}
\Pi_{1}=\frac{\partial L}{\partial\dot{\varphi}_{1}}=S\sin\varphi_{2}\,,\qquad\Pi_{2}=\frac{\partial L}{\partial\dot{\varphi}_{2}}=0\,,\qquad\Pi_{3}=\frac{\partial L}{\partial\dot{\varphi}_{3}}=S.\label{eq:MomentaSpositive}
\end{equation}
The last two terms define constraint equations. Therefore, if $\boldsymbol{\pi}=E_{0}\vb$,
the lagrangian (\ref{eq:Action_withS_frame}) can be written in canonical
form as
\begin{equation}
  \label{eq:Action_withS_Hamiltonian}
L_{\text{can}}=\boldsymbol{\Pi}\cdot\dot{\varphib}-E\dot{t}+\boldsymbol{\pi}\cdot\dot{\boldsymbol{x}}-N\left(E-E_{0}\right)-\eta_{2}\Pi_{2}-\eta_{3}\left(\Pi_{3}-S\right)
\end{equation}
where
$L_{\text{can}}[\bm{\varphi},\xb,t, 
\boldsymbol{\Pi},\boldsymbol{\pi},E,N, \eta_{2},\eta_{3}]$.

To highlight the physical degrees of freedom we solve the constraints,
impose the gauge condition $t=\tau$ and neglect boundary terms. The
lagrangian in the reduced phase space is given by
\begin{align}
  L_{\text{red}}\left[\varphi^{1}, \boldsymbol{x}, \Pi_{1},\boldsymbol{\pi}\right]=\Pi_{1}\dot{\varphi}^{1}+\boldsymbol{\pi}\cdot\dot{\boldsymbol{x}}-E_{0}.
\end{align}
Here $\varphi^{1}$, together with its conjugate $\Pi_{1}$, describe
the spin degrees of freedom. Note that the reduced action depends on
$8$ independent canonical variables, which coincides with the
dimension of the coadjoint orbit (see Table~\ref{tab:carrollions},
orbit \#2).

\subsection{Spinning vacuum action}
\label{subsec:Spinning-vacuum}

While the vacuum has a trivial action the spinning vacuum can be
obtained using the following representative of the coadjoint orbit
\begin{align}
  \label{eq:spin-vac-rep}
  \alpha=\left(J,\bzero,\bzero,0\right).
\end{align}
The resulting particle action can be obtained from the results of the
preceding section by setting $E_{0}=0$ in \eqref{eq:actionWithS2}, so
we will not repeat the computation. Let us note that the spinning
vacuum is described by the geometric action for
$SO(3)$~\cite{Alekseev:1988vx}.

\subsection{Massless carrollion action}
\label{sec:massless-with-spin}

Let us choose the following representative of the coadjoint orbit
\begin{align}
  \alpha=\left(J,\bzero,\boldsymbol{p},0\right),
\end{align}
where 
\begin{align}
  \|\boldsymbol{p}\|^{2}=p^{2},
\end{align}
with constant $p>0$. Using~\eqref{eq:action=general} we obtain the
lagrangian
\begin{align}
  L\left[\varphib,\boldsymbol{x}\right]=\frac{1}{2}\Tr\left[J^{T}R^{T}\frac{\partial R}{\partial\varphi^{a}}\right]\dot{\varphi}^{a}+\left(R\boldsymbol{p}\right)^{T}\dot{\boldsymbol{x}}
\end{align}
with nonzero Noether charges
\begin{align}
    \bj_{Q}&= R \bj  + \x \times R\p  &
    \p_{Q} &= R \pb.
\end{align}

Let us parametrise the rotation matrix as in~\eqref{eq:Coords_sphere}
and align $J$ and $\boldsymbol{p}$ with the $z$-axis, i.e.,
$J=h \varepsilon_{3}$ with helicity $h$ constant and
$\boldsymbol{p}=\left(0,0,p\right)$. The lagrangian is then given by
\begin{equation}
  \label{eq:Action_withS_frame-2}
  L \left[\boldsymbol{x},\varphib\right] = h 
  \left(\dot{\varphi}_{3} + \sin\left(\varphi_{2}\right) \dot
    {\varphi}_{1}\right) + p \,\boldsymbol{\hat{n}} \cdot
  \dot{\boldsymbol{x}},
\end{equation}
where the unit vector $\hat{\boldsymbol{n}}$ is defined
in~\eqref{eq:normal}, and the constants $h$ and $p$ are kept fixed in
the action. Using this parametrisation it is
straightforward to show that
\begin{align}
\|\p_{Q}\|^{2}&=p^{2} & \boldsymbol{j}_{Q}\cdot\boldsymbol{p}_{Q}=h  \|\p_{Q}\|,
\end{align}
are invariant quantities, in perfect agreement with our analysis in
Section~\ref{sec:massless-parallel}.

To write the massless carrollion action in canonical form we introduce
the canonical momenta
\begin{align}
  \boldsymbol{\pi}&=\frac{\partial L}{\partial\dot{\boldsymbol{x}}}=p\hat{\boldsymbol{n}} &
  \Pi_{1}&=\frac{\partial L}{\partial\dot{\varphi}_{1}}=h\sin\varphi_{2} &
  \Pi_{2}&=\frac{\partial L}{\partial\dot{\varphi}_{2}}=0 &
  \Pi_{3}&=\frac{\partial L}{\partial\dot{\varphi}_{3}}=h 
\end{align}
which satisfy the following constraints
\begin{align}
  \|\boldsymbol{\pi}\|^{2}-p^{2}&=0 &
  p\Pi_{1}-h \pi_{1}&=0 &
  \Pi_{2}&=0 &
  \Pi_{3}-h &=0.
\end{align}
Therefore, the lagrangian in canonical form becomes 
\begin{align}
  &L_{\text{can}}[\varphib,\boldsymbol{x},t,\boldsymbol{\Pi},\boldsymbol{\pi},E,N, \eta_{i}]= \\
  &\qquad
    \boldsymbol{\Pi}\cdot\dot{\varphib}
    +\boldsymbol{\pi}\cdot\dot{\boldsymbol{x}} -E\dot{t} -NE
    -\eta_{0}\left(\|\boldsymbol{\pi}\|^{2}-p^{2}\right)
    -\eta_{1}\left(p\Pi_{1}-h_{p}\pi_{1}\right)-\eta_{2}\Pi_{2}-\eta_{3}\left(\Pi_{3}-h\right)\, . \nonumber
\end{align}
Implementing the constraints $E=0$, $\Pi_{2}=0$ and $\Pi_{3}=h$,
up to boundary terms, the lagrangian takes the form
\[
  L_{\text{can}}[\varphi^{1},\boldsymbol{x},\Pi_{1},\boldsymbol{\pi}, \eta_{0},\eta_{1}]
  =\Pi_{1}\dot{\varphi}^{1}
  +\boldsymbol{\pi}\cdot\dot{\boldsymbol{x}}
  -\eta_{0}\left(\|\boldsymbol{\pi}\|^{2}-p^{2}\right)
  -\eta_{1}\left(p\Pi_{1}-h \pi_{1}\right).
\]

In the particular case when the spin vanishes, $h=0$, the
lagrangian \eqref{eq:Action_withS_frame-2} simplifies to
\begin{equation}
  L[\x]=p\boldsymbol{\hat{n}}\cdot\dot{\boldsymbol{x}},\label{eq:Action_withS_frame-2-1}
\end{equation}
while the lagrangian in canonical form becomes
\begin{align}
  L_{\text{can}}[\boldsymbol{x},t,\boldsymbol{\pi}, E,N,\eta_{0}]
  =\boldsymbol{\pi}\cdot\dot{\boldsymbol{x}}-E\dot{t}-NE
  -\eta_{0}\left(\|\boldsymbol{\pi}\|^{2}-p^{2}\right)\, .
\end{align}
The action for the case $h=0$ was previously studied
in~\cite{deBoer:2021jej}.

As discussed in Appendix~\ref{sec:orbits-mod-autos}, under the outer
automorphisms of the Carroll algebra discussed in
Appendix~\ref{sec:automorphisms}, the orbit discussed in this section
is mapped to those of cases \#5 and \#7 (see Table
\ref{tab:carrollions}). As a consequence, the action for all these
cases takes precisely the same intrinsic form. However, the physical
interpretation may differ in general. We shall revisit this point in
Section~\ref{sec:carr-vers-fract} where we will provide a useful
physical interpretation in the context of fractons and also discuss
the mobility restrictions.

\subsection{Generic massless carrollion action}
\label{sec:generic-1}

The action for the generic massless case can be constructed using
the following representative element of the coadjoint orbit
\begin{align}
  \label{eq:L-generic}
  \alpha=\left(0,\boldsymbol{k},\boldsymbol{p},0\right),
\end{align}
where $\k$ and $\p$ are nonzero and not parallel.

From~\eqref{eq:action=general} we obtain the lagrangian
\begin{align*}
L\left[\boldsymbol{\mathbf{\varphi}},\boldsymbol{x},\vb\right] & =\left(R\boldsymbol{p}\right)\cdot \dot{\boldsymbol{x}}+\left(R\boldsymbol{k}\right) \cdot \dot{\vb}.
\end{align*}
The non-vanishing Carroll conserved charges of this action are given by
\begin{align}
    \bj_{Q}&= \bv \times R\k + \x \times R\p &
    \k_{Q} &= R\k &
    \p_{Q} &= R \pb
\end{align}
From these expressions it is clear that we indeed recover 
\begin{align}
\|\p_{Q} \|^{2}&=p^{2}  & \| \vb_{Q} \|^{2} &=k^{2} & \pb_{Q} \cdot \vb_{Q}=pk\cos\theta,  
\end{align}
as discussed in Section~\ref{sec:generic}.

Since the action is more transparent in canonical form we introduce
the canonical momenta are given by
\begin{align}
  \boldsymbol{\pi}=\frac{\partial L}{\partial\dot{\boldsymbol{x}}}&=R\boldsymbol{p} & \boldsymbol{\pi}_{v}&=\frac{\partial L}{\partial\dot{\vb}}=R\boldsymbol{k}
\end{align}
which are restricted to obey the constraints
\begin{align}
\|\boldsymbol{\pi}\|^{2}-p^{2}&=0 & \|\boldsymbol{\pi}_{v}\|^{2} -k^{2}&=0 & \boldsymbol{\pi}\cdot\boldsymbol{\pi}_{v}-pk\cos\theta&=0.
\end{align}
Therefore, the lagrangian in canonical form can be written as
\begin{equation}
  \label{eq:Icangeneral}
  L_{\text{can}}
  =-E\dot{t}+\boldsymbol{\pi}\cdot\dot{\boldsymbol{x}}+\boldsymbol{\pi}_{v}\cdot\dot{\vb}
  -NE
  -\eta_{1}\left(\|\boldsymbol{\pi}\|^{2}-p^{2}\right)
  -\eta_{2}\left(\|\boldsymbol{\pi}_{v}\|^{2}-k^{2}\right)
  -\eta_{3}\left(\boldsymbol{\pi}\cdot\boldsymbol{\pi}_{v}-pk\cos\theta\right).
\end{equation}
where $L_{\text{can}}[\vb,\xb,t,\bm{\pi}_{v},\bm{\pi},E,N,\eta_{i}]$.
In addition to the position $\boldsymbol{x}$ and its conjugate
momentum $\boldsymbol{\pi}$, the state of a generic massless Carroll
particle is also described by a vector field $\vb$ and its
conjugate momentum $\boldsymbol{\pi}_{v}$, which can be considered as
an additional internal degree of freedom. In
Section~\ref{sec:carr-vers-fract}, we shall provide a physical
interpretation of this case in the context of fractons and provide a
discussion of the solutions. The counting of degrees of freedom shows
that there are 6 independent canonical variables, in agreement with
the dimension of the orbit.

\section{From Poincaré to Carroll particles}
\label{sec:carr-vs-massl}

In this section we study the carrollian limit of Poincaré particles.
We intend to provide a useful orientation and not a comprehensive
study of the possible limits. In short, there is a clear relation
between massive spinning Poincaré and massive spinning Carroll
particles.  Massless carrollions seem to derive from a limit of
Poincaré tachyons, and the limit of massless Poincaré particles seems
to trivialise, see Figure~\ref{fig:mom_lim}.

\begin{figure}
\centering
\begin{tikzpicture}
  \node at (0,0) {\includegraphics[width=0.5 \textwidth]{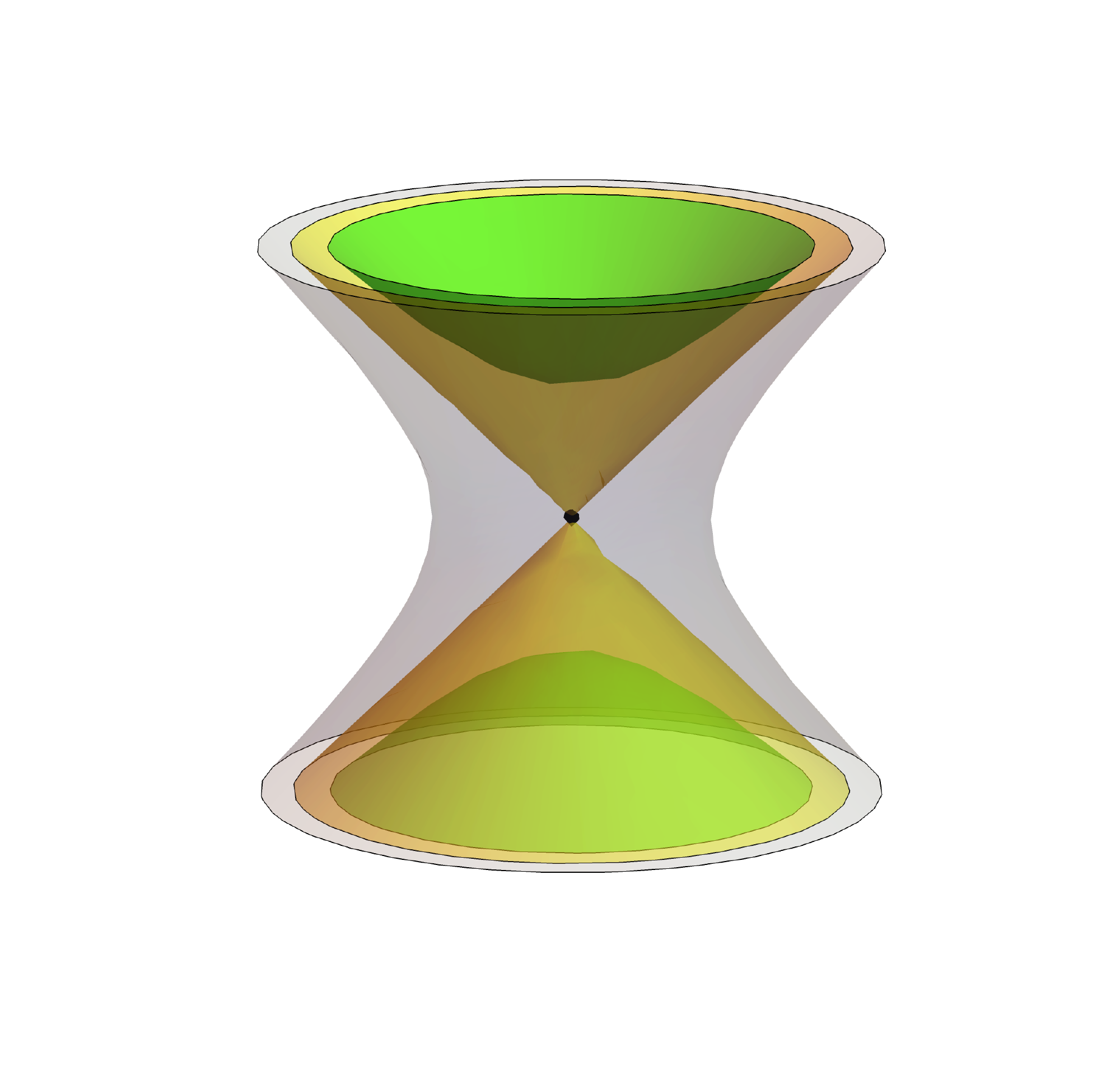}};
  \node at (7,0) {\includegraphics[width=0.5 \textwidth]{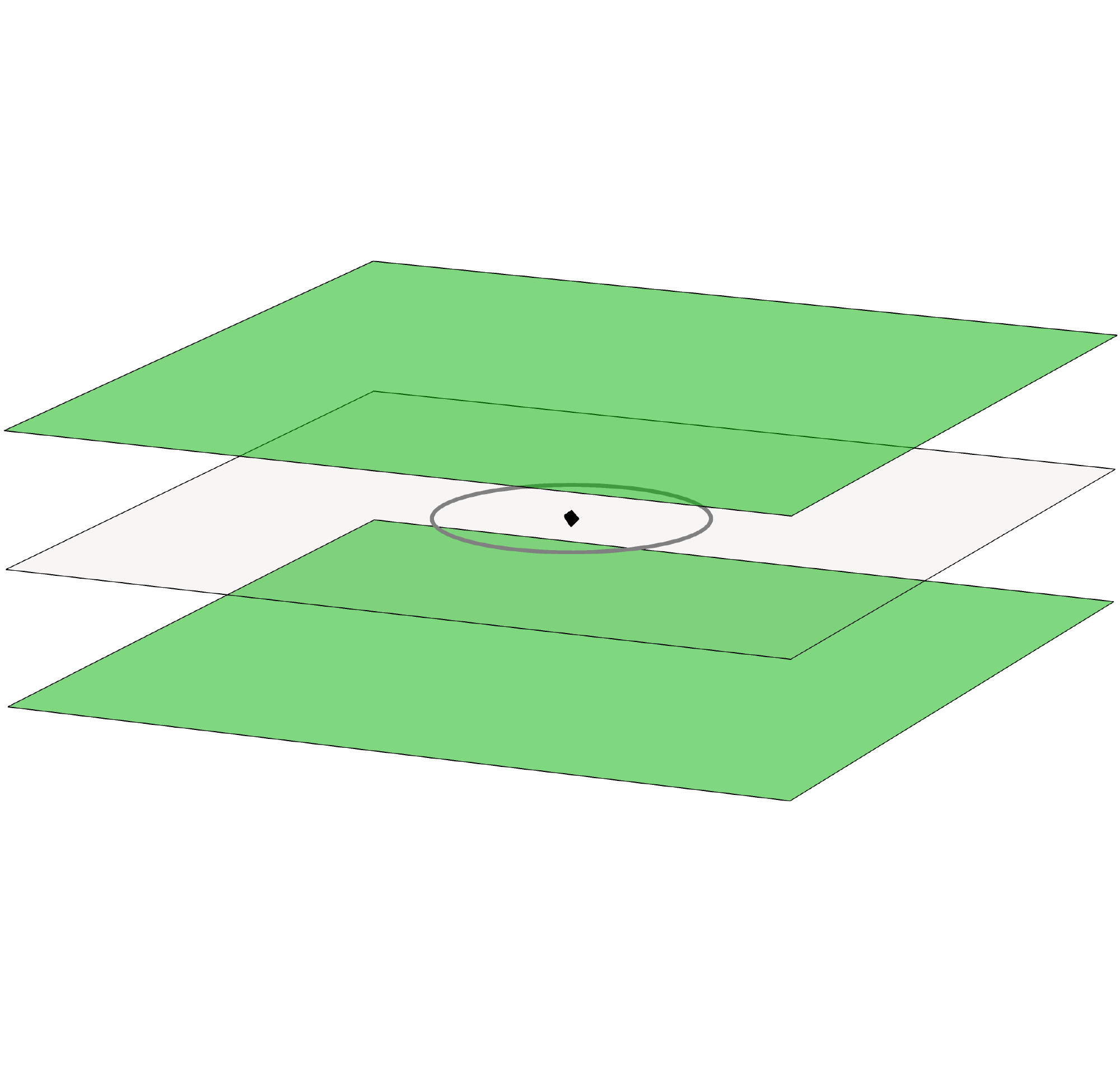}};
  \node at (0,-3) {Poincaré};
  \node at (7,-3) {Carroll};
  \draw[-Stealth,thick] (2.7,-3) -- (2.7,3);
  \node[right] at (2.7,2.9) {$E$};
\end{tikzpicture}
\caption{This figure shows the coadjoint orbits of the Poincaré and
  Carroll group in momentum space $(E,\p)$. The Poincaré orbits are
  given by two families of massive particles (green), two massless
  orbits (yellow),
  one family of tachyonic orbits (gray) and the vacuum (black dot).\\
  In the Carroll limit the green massive Poincaré orbits
  (hyperboloids) flatten out and lead to the planes (see
  Section~\ref{sec:massive-case}). The lightcone would lead to a the
  $E=0$ plane, but since $E=\|\pb \|$ this orbit actually vanishes in
  the limit. The sphere ($\bm{p}^{2} = const.$), represented as a
  circle in the $E=0$ plane on the right-hand side, can be seen to
  arise as a limit of the tachyonic orbits (see
  Section~\ref{sec:massless-case} for details). The whole carrollian
  $E=0$ plane is foliated by such spheres.
  \\
  Let us emphasise that this figure only represents the $(E,\p)$ part
  of the full dual space $(\jb,\bv,\p,E)$ and the complete structure
  of the orbits is more intricate and involves spin degrees of
  freedom.}
    \label{fig:mom_lim}
\end{figure}

\subsection{Poincaré particles}
\label{sec:poincare-particles}
  
We first review some well-known facts about Poincaré particles. The
coadjoint action of a Poincare transformation $(\Lambda,a)$ on the
four-momenta $p^{\mu} = (E,\p)$ is given by $p' = \Lambda p$, where we
can see that the translations act trivially. It follows that
$p_{\mu}p^{\mu}$ is an invariant on the coadjoint orbits. Using the
rest mass $M$ this foliates the momentum space into disjoint orbits
\begin{align}
  \label{eq:psqared}
  -  p_{\mu}p^{\mu}= E^{2} -  \|\p\|^{2}  =\left\{
    \begin{array}{cl}
      M^{2} & \quad\text{massive} \\
      0 & \quad\text{massless} \\
      - M^{2} & \quad\text{tachyonic}
    \end{array}
\right. \, .
\end{align}
They are given by (see Figure~\ref{fig:mom_lim}): a family of
two-sheeted hyperboloids, with each sheet being an orbit,
corresponding to the massive momenta, one family of one-sheeted
hyperboloids, corresponding to tachyonic momenta and the future and
past deleted lightcone, corresponding to massless momenta, and the
origin corresponding to the vacuum.

In the following we will also use the second Poincaré invariant, the
square of the Pauli--Luba\'nski vector
\begin{align}
  \label{eq:PLv}
 W^{2}  =-  (\jb \cdot \p)^{2} + \| E \jb + \p \times \k \|^{2} \, ,
\end{align}
which should be contrasted with its carrollian
analog~\eqref{eq:carr-inv}.

\subsection{Carroll limits}
\label{sec:carr-limits}

Let us now discuss the Carroll limit of the Poincaré coadjoint orbits.
We do not claim that these are the unique limits and we will call the
contraction parameters uniformly by $c$ (which will not necessarily
have the units of velocity).

\subsubsection{Massive Poincaré to massive Carroll particles}
\label{sec:massive-case}

Upon suitable rescalings the carrollian limit of the invariants of the
Poincaré group lead to
\begin{align}
  \label{eq:psqtozero}
 \lim_{c \to 0} (-p_{\mu}p^{\mu}) = \lim_{c \to 0} (E^{2} - c^{2} \|\p\|^{2})  =E^{2} =M^{2}
\end{align}
and
\begin{align}
  \label{eq:Wtocarr}
  \lim_{c \to 0} W^{2}
  = \lim_{c \to 0} \left[- c^{2} (\jb \cdot \p)^{2} + \| E \jb + \p \times \k\|^{2}\right]
  =  \|E \jb + \p \times \k\|^{2} = E^{2} S^{2} \, .
\end{align}
When we identify $M \to E_{0}$ we recover the invariants of the
massive carrollions that were already discussed in
Section~\ref{sec:e-neq-0}. Since they uniquely characterise the
massive orbits of Poincaré and Carroll we can map for any spin $S$ the
massive Poincaré particles to massive Carroll particles and vice
versa.

Geometrically the Carroll limit of~\eqref{eq:psqtozero} implies that
in momentum space $(E,\p)$ the hyperbolic mass-shells get flattened to
planes with $E=\pm E_{0}$, cf.,~the green surfaces in
Figure~\ref{fig:mom_lim}. The stabiliser of the momentum orbit is in
both cases given by $SO(3)$, which implies a close relation between
the induced representations.

\subsubsection{Carroll limit of the vacuum sector}
\label{sec:vacuum-limit}

If we restrict to the Poincaré vacuum sector $p^{\mu} = 0$ the action
of the translations on the coadjoint orbits trivialises and we are
left with the coadjoint orbits of the Lorentz group $SO(3,1)$. In this
sector we have the invariants $- \|\bm{k}\|^{2} + \|\bm{j}\|^{2}$ and
$\bm{j} \cdot \bm{k}$ (they derive from $j^{\mu\nu}j_{\mu\nu}$ and
$ \epsilon^{\mu\nu\rho\xi}j_{\mu\nu} j_{\rho\xi}$, respectively).

Taking the limit
\begin{align}
  \label{eq:klimit}
  \lim_{c\to 0}(\| \k\| ^{2} - c^{2} \|\jb\|^{2}) =\|\k\|^{2} \, ,
\end{align}
and leaving the other invariant unaltered provides a limit to what we
called centrons in Table~\ref{tab:carrollions}. So they derive from
the Poincaré vacuum
sector.

Taking the dual limit of~\eqref{eq:klimit} leads to the
invariants of the spinning vacua $\|\bm{j}\|^{2}=j^{2}$.

\subsubsection{Carroll limit of massless and tachyonic Poincaré particles}
\label{sec:massless-case}

For the remaining orbits we restrict the discussion to energy and
momentum space. One can then take the dual limit
of~\eqref{eq:psqtozero}, upon suitable rescaling, given by
\begin{align}
  \label{eq:psqtozero2}
  \lim_{c \to 0} (p_{\mu}p^{\mu}) = \lim_{c \to 0} (-c^{2} E^{2} +  \|\p\|^{2})  = \|\p\|^{2} =M^{2} \,.
\end{align}
When we identify $M \to p$ we recover the relation $\| \p\| = p$ which
was already discussed in Section~\ref{sec:carr-part-3+1}. This limit
is just the restriction of the three dimensional de Sitter space to
the two spheres of radius $p$ in the $E=0$ plane, see the gray
hyperboloid and circle in Figure~\ref{fig:mom_lim}. This suggests that
the massless carrollions emerge from tachyons.

From this perspective is it also clear that the massless Poincaré
particles lead to $E=0$ and $\bm{p} = \zerob$. We start by setting
$M=0$ in~\eqref{eq:psqared}, which leads to the lightcone
\begin{align}
  \label{eq:massless}
  E^{2} = \|\p\|^{2}  \, .
\end{align}
If we want to keep rotational invariance we can send either side to
zero, but that implies that both sides vanish. Geometrically we can
think about it as light cone that opens up with the limit being the
$E=0$ plane. But due to the relation~\eqref{eq:massless} now $\pb$
also vanishes.

\section{Carroll/fracton correspondence}
\label{sec:carr-vers-fract}

In this section we explore the relationship between the Carroll and
dipole symmetries. In particular, we shall focus on models with a
dipole symmetry. As it was explained, their symmetries are closely related, 
and therefore one can map carrollions to fractons using the
Carroll/fracton correspondence. Fracton monopoles are related to massive
carrollions, both of which cannot move. On the other hand, the massless
carrollions are related to fractonic dipoles. Some of the unusual
properties of carrollions and fractons will be elucidated by an analysis
based on the use of two fundamental monopoles, together with the
application of appropriate constraints on the phase space.

\subsection{Fracton symmetry and Carroll/fracton correspondence}
\label{sec:fracton-symmetry}

From a purely theoretical perspective, fractons are very puzzling
objects, as they seem to defy the standard methods of quantum field
theory (see the reviews
\cite{Nandkishore:2018sel,Pretko:2020cko,Grosvenor:2021hkn} for
(original) references, further applications and details). In
particular, for the complex scalar field theory introduced in
\cite{Pretko:2018jbi}, there exists a conserved dipole charge in
addition to the electric charge. Suppose these symmetries act as
\begin{align}
  \phi(t,\xb) \to e^{i (\alpha + \bm{\vb}\cdot\x)} \phi(t,\xb) \, ,
\end{align}
where $\alpha$ and $\vb$ parametrise the charge and dipole
transformations, respectively. When we then insist on spatial
derivatives in the action, terms like
$\partial_{i}\phi^{*}\partial_{i}\phi$ are forbidden and we are led to
non-gaussian and non-lorentzian theories~\cite{Bidussi:2021nmp}
(see~\eqref{eq:frac-field} for an explicit action). An immediate
consequence is that the conventional expansion of the fields in terms
of oscillators, commonly used in quantum field theory to define
particle states, is not directly applicable to fractonic models.
Therefore, a natural unanswered question is how to determine the
elementary excitations or particles of fractonic theories.

To answer this question we can use a remarkable
fact~\cite{Bidussi:2021nmp} (see also~\cite{Marsot:2022imf}):
\emph{the symmetries of fractonic theories with a conserved electric
  and dipole charges coincide with the Carroll symmetries, up to the
  inclusion of an additional central element}. Thus,
the analysis of the coadjoint orbits developed for carrollian theories in the previous
sections can be easily extended to characterise the elementary
excitations of fractonic models with dipole charges. While they have
different interpretations it is interesting to think about carrollions
from the fracton perspective and vice versa.

The generators of the dipole algebra~\cite{Gromov:2018nbv}, spanned
by the fracton energy $H_{F}$, angular momentum $L_{ab}$, linear momentum
$P_{a}$, dipole charge $D_{a}$, and electric charge $Q$, have the
following non-vanishing commutation relations
\begin{equation}
  \label{eq:frac-alg}
  \begin{split}
  \left[L_{ab},L_{cd}\right] & =\delta_{bc}L_{ad}-\delta_{ac}L_{bd}-\delta_{bd}L_{ac}+\delta_{bd}L_{ac} \\
  \left[L_{ab},P_{c}\right] & =\delta_{bc}P_{a}-\delta_{ac}P_{b}\\
  \left[L_{ab},D_{c}\right] & =\delta_{bc}D_{a}-\delta_{ac}D_{b}\\
  \left[D_{a},P_{b}\right] & =Q\delta_{ab} 
  \end{split}\, .
\end{equation}
This is precisely the Carroll algebra~\eqref{eq:carroll-kla-brackets},
when we apply $Q \mapsto H$ and $D_{a} \mapsto B_{a}$ and ignore the
additional generator $H_{F}$ that commutes with all the generators of the
algebra. This implies that in a Carroll/fracton correspondence the Carroll
energy is in correspondence with the fractonic electric charge, and
the carrollian center-of-mass is related to the dipole charge. The
fracton energy, which commutes with all the generators, does not have
a counterpart in the Carroll algebra. Table~\ref{tab:carrvsfrac}
exhibits the precise correspondence between the conserved quantities
of carrollian and fractonic theories. In~\cite{IP:2023} we will
provide a study of other fracton-like models and provide their
correspondence with different non-relativistic systems.

\begin{table}
  \centering
  \caption{Correspondence of conserved quantities between carrollian
    and fractonic theories}
   \resizebox{0.6 \linewidth}{!}{
     \begin{tabular}{l  l }
       \toprule
       Carroll particles  \qquad \qquad \qquad \qquad      & Fractonic particles          \\ \midrule
       angular momentum $\jb$  & angular momentum $\jb$\\
       center-of-mass $\k$    & dipole moment $\db$   \\
       momentum  $\pb$       & momentum  $\pb$      \\
       energy  $E$         & charge      $q$     \\
       ---              & energy       $E$    \\
       \bottomrule
  \end{tabular}
  }
  \label{tab:carrvsfrac}
\end{table}

While the Carroll/fracton correspondence might be useful in general we will
now apply it on the level of the elementary particles. In what
follows, we shall use the results of the previous sections to identify
the fracton particles. Like carrollions, fractons fall into two
categories: when the charge $q$ vanishes they correspond to (immobile)
monopoles, but when the charge is nonzero they are given by neutral
fractons, in particular dipoles which are mobile.

This can be seen by applying the correspondence of Table~\ref{tab:carrvsfrac}
to the coadjoint action~\eqref{eq:coadjoint-rep-n=3-summary}. The
coadjoint action for fractons acts therefore on
$\alpha =(\bj,\db,\p,q,E)$, where the conserved quantities are
described in Table~\ref{tab:carrvsfrac}, via the coadjoint
$\Ad^*_g \alpha = (\bj', \db', \p', q',E')$ as follows
\begin{equation}
  \label{eq:coadjoint-fractons}
  \begin{split}
    \bj' &= R \bj + \bv \times R\db + \ab \times R\p + q\vb \times \ab\\
    \db' &= R\db + q\ab\\
    \p' &= R\p - q\vb\\
    q' &= q \\
    E' &=E
  \end{split} \, .
\end{equation}
The group element $g = (R,\bv,\ba,\lambda,s)$ is given by rotations $R$,
dipole transformations $\bv$, spatial translations $\ab$, charge
rotations $\lambda$ and temporal translations $s$ (where the last two
act trivially). A few remarks concerning the fracton coadjoint orbits
are in order:
\begin{itemize}
\item By inspecting~\eqref{eq:coadjoint-fractons} and by our
  discussions of the Carroll symmetries it is clear that $E$, $q$ and
  $W^2 = \|q \bj + \p \times \db \|^2$ are invariants.

\item The coadjoint action indeed agrees with the intuition of the
  transformation behaviour of the well-known dipole moment
  $\db=\int \x \rho(\x) d^{3}x$ of electrodynamics ($\rho$ is the
  charge density). Suppose we have a particle with nonzero charge $q$
  at the origin, then $\rho(\x)=q \delta(\x)$ and $\db=\zerob$. If we
  now shift the particle to the point $\ab$ we obtain
  $\rho(\x)=q \delta(\x-\ab)$ and as expected a non-vanishing dipole
  moment given by $\db'=q \ab$. This is just a manifestation of the
  well-known fact that the dipole moment for nonzero charge depends on
  the choice of origin, see, e.g,~\cite{griffiths1999introduction}.

  This shift is also precisely the one we obtain by the coadjoint
  action of a pure translation $\ab$, with group element
  $g = (R=1,\bv=\zerob,\ba,\lambda=0,s=0)$, on a charged particle with
  zero dipole moment $\alpha=(\jb=\zerob,\db=\zerob,\p=0,q,E)$,
  see~\eqref{eq:coadjoint-fractons}.

\item A subtle difference between Carroll and fracton particles is the
  global group structure.  While we assumed that the action of Carroll
  time translations results in noncompact orbits (so that the subgroup
  generated by $H$ is isomorphic to $\RR$), the group of charge phase
  rotations is expected to be $U(1)$.  This generator is however
  central and therefore its coadjoint action is trivial and hence, at
  the level of the coadjoint orbits, the topology of the subgroup it
  generates is immaterial.
\end{itemize}

The elementary fractonic particles are physically interpreted as
charged monopoles and different classes of neutral fractons, in
particular dipoles. Their dynamics is described by actions that are in
a one-to-one correspondence with those of
Section~\ref{sec:particle-actions}. In particular, the analysis of the
dynamics of these systems will reveal the restricted mobility of
fractonic particles. For simplicity in the presentation, we shall only
consider the cases of spinless particles. The actions for spinning
fractons can be directly derived by using the results of
Section~\ref{sec:particle-actions}. So the fractonic dual
of~\eqref{eq:action=general} and the starting point for our analysis
of the actions $S=\int L d\tau $ is
\begin{align}
  \label{eq:frac-action-general}
    L[R(\varphib),\bv,\x,\phi,t]  
  &= \tfrac12 \Tr \left( J^T R^T \dot R  \right)
  + (R\db)\cdot \dot\bv
  + (R\p)\cdot  \dot\x
  + q 
    \left(
    \dot \Phi
  +  \tfrac12  \x \cdot \dot\bv
    -\tfrac12  \bv \cdot \dot\x
    \right)
    - E \dot t
 \nonumber   \\
  &= \tfrac12 \Tr 
    \left(
    J^T R^T \dot R
    \right)
  + (R\db + \tfrac12 q \x) \cdot \dot\bv
  + (R\p -\tfrac12 q \bv)\cdot \dot\x
    + q \dot \Phi
    - E \dot t
\end{align}
where the canonical pair $\left(q,\Phi\right)$ represents the total
electric charge $q$ and its canonical conjugate $\Phi$.

\subsection{Fractonic monopole ($q\protect\neq0$).}
\label{sec:fract-monop}

Let us consider a fractonic system with non-vanishing total electric
charge. This means an element of the dual space of the form
$\left(\boldsymbol{j},\boldsymbol{d},\boldsymbol{p},E,q\right)$, where
everything vanishes except $E=E_{0}$ and $q=e$, for constants
$E_{0}\in\mathbb{R}$ and $e\neq0$ (this corresponds to massive
spinless carrollions).

The particle lagrangian associated to this orbit can be directly
obtained from the results in Section~\ref{sec:action-1}. For fractons
this case can be physically interpreted as the description of the
dynamics of an elementary monopole carrying a total electric charge of
magnitude $e$.
The lagrangian in canonical form can be written as 
\begin{equation}
  \label{eq:LFracton_monopole}
  L_{\text{can}}[\x,q,t,\bm{\pi},\Phi, E, N,\eta]=-E\dot{t}+\boldsymbol{\pi}\cdot\dot{\boldsymbol{x}}+\Phi\dot{q}-N\left(E-E_{0}\right)-\eta\left(q-e\right).
\end{equation}
While the Lagrange multipliers $N$ and $\eta$ enforce that the
energy is $E_{0}$ and the charge $e$, the variation with respect to
$\bm{\pi}$ leads to
\begin{align}
  \label{eq:xdot-zero}
  \dot{\boldsymbol{x}}=\zerob \, . 
\end{align}
Therefore we recover the characteristic feature that a single
fractonic monopole cannot move, which is the now a consequence of the
action~\eqref{eq:LFracton_monopole}.

The non-vanishing conserved charges that realise the fractonic algebra
as Poisson brackets are given by the charges
\begin{align}
  \boldsymbol{j}_{Q}&=\boldsymbol{x}\times\boldsymbol{\pi} &
  \boldsymbol{d}_{Q}&=e\boldsymbol{x} &
  \boldsymbol{p}_{Q} &=\boldsymbol{\pi} &
  q_{Q}&=e & 
  E_{Q}&=E_{0}                                                
\end{align}
where $\boldsymbol{j}_{Q}$ is the total angular momentum,
$\boldsymbol{p}_{Q}$ the linear momentum, $q_{Q}$ the total electric
charge, $\boldsymbol{d}_{Q}$ the dipole charge and $E_{Q}$ the
fractonic energy. As expected and due to~\eqref{eq:xdot-zero} the
dipole charge is indeed conserved $\dot \db = \zerob$. We also recover
the characteristic commutation relations
\begin{align}
  \label{eq:dp-can}
  \{ \db_{Q}, \p_{Q}\} = q_{Q}\bm{1} \, .
\end{align}

The lagrangian defined on the reduced phase space can be derived by
solving the constraints and imposing the gauge fixing $t=\tau$.
Neglecting boundary terms one finds
\begin{equation}
  \label{eq:LredFracton_monopole}
  L_{\text{red}}[\x,\bm{\pi}]=\boldsymbol{\pi}\cdot\dot{\boldsymbol{x}}-E_{0} = \frac{1}{e}\pb_{Q}\cdot \dot \db_{Q}-E_{0}
\end{equation}
which again manifests the canonical relation between linear momentum
and dipole moment~\eqref{eq:dp-can} and shows the equivalence to the
massless carrollion upon correspondence.

\subsection{\texorpdfstring{Fractonic dipole ($q=0$, $d \protect\neq0$, $\p=\bzero$)}{Fractonic dipole (q=0, d neq 0,p=0)}}
\label{sec:fracton-dipole-ele}

In analogy to the electromagnetic theory, an elementary dipole has a
vanishing electric charge, and a nonzero dipole moment. Consequently,
this fractonic excitation can be described by massless carrollions
(with zero helicity for simplicity, see Table~\ref{tab:carrollions}).

The lagrangian can be written as
\begin{equation}
  \label{eq:L_fracton_8_reduced-1-2}
  L_{\mathrm{can}}[\db,t,\vb,E,N,\eta]=-E\dot{t}+\db\cdot\dot{\vb}-N\left(E-E_{0}\right)-\eta\left(\|\db\|^{2}-d^{2}\right).
\end{equation}
Here $\db$ is the dipole moment, and can be considered as a
fundamental degree of freedom which is varied in the action. The
constraint enforced by $\eta$ tells us that the constant $d$
(which is not varied) fixes the magnitude of the dipole moment.
Solving the constraint $E-E_{0}$ and imposing the gauge condition
$t=\tau$, the lagrangian can be rewritten as
\begin{equation}
  \label{eq:L_fracton_8_reduced-1-2-1}
  L_{\mathrm{can}}[\db,\vb,\eta]=\db\cdot\dot{\vb}-E_{0}-\eta\left(\|\db\|^{2}-d^{2}\right).
\end{equation}
The position of the dipole is not a dynamical variable in the
lagrangian due to the vanishing of its conjugate momentum, the linear
momentum. Consequently, the action does not specify the position of
the dipole in space, and as a result, there are no mobility
restrictions for this pure dipole. This action should be compared with
the massless carrollions in Section~\ref{sec:massless-with-spin} to
which the correspond to.

To make this even more manifest it is useful to remember that the
original action indeed had a dependence on $\xb$ which dropped out
since we looked at orbits with $\pb=\zerob$. If we do not integrate it
out, as we did earlier, but calculate its canonical momentum
$\bm{\pi}=\pd L / \pd \dot\xb=0 $ we find the constraint that the
canonical momentum vanishes. The action is then given by
\begin{equation}
  \label{eq:L_dipolewithx}
  L_{\mathrm{can}}[\db,\xb, \vb, \bm{\pi},\eta,\bm{u}]=\db\cdot\dot{\vb} + \bm{\pi} \cdot \dot \x -E_{0}-\eta\left(\|\db \|^{2}-d^{2}\right)  -\bm{u} \cdot \bm{\pi}
\end{equation}
where $\bm{u}$ are Lagrange multipliers that enforce the constraint
that $\bm{\pi}$ vanishes. The corresponding gauge transformations tell
us that $\xb$ is arbitrary, i.e., the position of the dipoles is not
restricted. We can also see this from the equation of motion coming
from the variation of $\bm{\pi}$ which leads to
\begin{align}
  \dot \x = \bm{u}
\end{align}
where the Lagrange multiplier $\bm{u}$ is an arbitrary function of
time. So again there is no restriction on the position, in drastic
contradistinction to the monopoles.

\subsection{Generic fractonic dipole ($q=0$, $d\protect\neq0$,
  $p\protect\neq0$)}
\label{sec:gener-fract-dipole}

Let us consider the lagrangian associated with generic dipoles
(corresponding to generic massless carrollions in
Table~\ref{tab:carrollions})
\begin{align}
\label{eq:Icangeneral-2-1}  
L_{\text{can}}=&-E\dot{t}+\p\cdot\dot{\boldsymbol{x}}+\db\cdot\dot{\vb}+\Phi\dot{q}-N\left(E-E_{0}\right)-\eta_{0}q-\eta_{1}\left(\|\p\|^{2}-p^{2}\right) \nonumber \\
&-\eta_{2}\left(\|\db\|^{2}-d^{2}\right)
-\eta_{3}\left(\p\cdot\db-pd\cos\theta\right) \, ,
\end{align}
where $L_{\mathrm{can}}[\vb,\xb,q,t,\db,\p,\Phi,E,N,\eta_{i}]$.

The total electric charge vanishes for this orbit. Therefore, this
case can be interpreted as a dipole $\db$ with an additional degree of
freedom that describes its position in space $\x$.

The non-vanishing charges of the dipole algebra are $\p_{Q} =\p$,
$\db_{Q}=\db$ together with
\begin{align}
\boldsymbol{j}_{Q}&=\boldsymbol{x}\times\p+\vb\times\db &  E_{Q}&=E_{0} \, .
\end{align}
After solving the trivial constraints $q=0$, $E-E_{0}=0$, and imposing
the gauge fixing condition $t=\tau$, the lagrangian can be rewritten
in the following form
\begin{equation}
  \label{eq:Icangeneral-2-1-1}
  L_{\mathrm{red}}=\p\cdot\dot{\boldsymbol{x}}+\db\cdot\dot{\vb}-E_{0}
  -\eta_{1}\left(\|\p\|^{2}-p^{2}\right)-\eta_{2}\left(\|\db\|^{2}-d^{2}\right)-\eta_{3}\left(\p\cdot\db-p d\cos\theta\right).
\end{equation}
where $L_{\mathrm{red}}[\vb,\xb,\db,\p,N,\eta_{i}]$.

Before we discuss solutions of this model, let us provide a physical
interpretation of the action \eqref{eq:Icangeneral-2-1-1}. To that end
it is useful to construct a dipole by considering two monopoles with
opposite charges separated by a small distance. This precisely
corresponds to the notion of an ideal dipole in electrodynamics. This
system, by definition, would not be considered an elementary object,
but rather a composite one. However, once the constraints
\begin{align}
  \label{eq:const8fractons}
 \|\p\|^{2} &= p^{2}& \|\db\|^{2} &=d^{2} & \p\cdot\db&=p d\cos\theta 
\end{align}
are imposed, the system can be regarded as a fundamental one.

\subsection{Dipoles from two monopoles}
\label{sec:dipoles-from-two}

To obtain a clear physical interpretation of the above systems it is
instructive to construct a dipole in terms of two
monopoles.\footnote{In the carrollian case, a similar construction can
  be performed, although instead of a dipole created from two
  monopoles of opposite charges, one must instead consider a
  particle/antiparticle pair with energies of the same magnitude but
  opposite sign.} Thus, let us consider two monopoles with opposite
charges of equal magnitude $e$ and $-e$. The action for the generic
dipole is constructed by taking the sum of the actions for each
individual monopole and subsequently imposing the
constraints~\eqref{eq:const8fractons}.

As a starting point, let us consider the lagrangian
\eqref{eq:LFracton_monopole}, describing two monopoles
\begin{align}
  L[\xb_{i},\bm{\pi}_{i}]=\boldsymbol{\pi}_{1}\cdot\dot{\boldsymbol{x}}_{1}+\boldsymbol{\pi}_{2}\cdot\dot{\boldsymbol{x}}_{2}-\left(E_{1}+E_{2}\right).
\end{align}
The total dipole moment, linear momentum and energy of the system are
given by
\begin{align}
\label{eq:charges_8_Fracton}  
  \db&=e\left(\boldsymbol{x}_{1}-\boldsymbol{x}_{2}\right)&
  & \pb=\boldsymbol{\pi}_{1}+\boldsymbol{\pi}_{2} &
    E&=E_{1}+E_{2}.
\end{align}
As discussed previously, this system cannot be considered as
fundamental without imposing the
constraints~\eqref{eq:const8fractons}. This system is at this stage
not elementary; that is, it is analogous to a reducible, rather than
to an irreducible, representation.

Introducing the following quantities
\begin{align}
  \boldsymbol{x}&:=\frac{1}{2}\left(\boldsymbol{x}_{1}+\boldsymbol{x}_{2}\right) &
  \vb&:=\frac{1}{2e}\left(\boldsymbol{\pi}_{2}-\boldsymbol{\pi}_{1}\right),  
\end{align}
and imposing the constraints (\ref{eq:const8fractons}), we obtain
precisely the lagrangian of the generic fractonic
dipole~\eqref{eq:Icangeneral-2-1-1}. The degrees of freedom of this
system are characterised by the position $\boldsymbol{x}$ of the
dipole, which is given by the average of the positions of each
monopole, and the total dipole moment $\db$, as well as their
respective canonical conjugate momenta $\pb$ and $\vb$. These
variables are subject to the constraints in \eqref{eq:const8fractons}.
Out of two elementary monopoles, which together are reducible, we have
thus created an elementary fundamental system, the generic fractonic
dipole.

The equations of motion imply the conservation of the momentum and
dipole moment, $\dot\pb=\dot{\db}=\zerob$. In
addition, one has
\begin{align}
  \label{eq:EOM_Frac_8} 
  \dot{\x}&=2\eta_{1}\p+\eta_{3}\db &
   \dot{\vb}&=2\eta_{2}\db+\eta_{3}\p \, . 
\end{align}
In analogy with the gauge fixing used in~\cite{deBoer:2021jej}, let us impose the following conditions:
\begin{align}
  \|\dot{\boldsymbol{x}}\|-p&=0 & \|\dot{\vb}\|-d&=0 \, .
\end{align}
The equations of motion lead to $\eta_{1}=\eta_{2}=1/2$ and $\eta_{3}=0$.
Therefore, they can be integrated as follows
\begin{align}
  \label{eq:dipole-inte}
  \boldsymbol{x}\left(t\right)=\p t+\boldsymbol{x}_{0}\,,\qquad\vb\left(t\right)=\db t +\vb_{0}.
\end{align}
Again, it is evident that the elementary dipole is not constrained to
remain static, in contrast to the fractonic monopole. However, in
contradistinction to the dipole of
Section~\ref{sec:fracton-dipole-ele}, it can be seen
from~\eqref{eq:dipole-inte} that the spatial evolution is not
completely undetermined anymore.

This realisation of the generic dipole in terms of two monopoles also
elucidates some of the particles we have already discussed.

\subsubsection{Dipole moment as a degree of freedom}
\label{sec:derivation-case-5}

Let us consider the particular case when $\p=\zerob$. According to
\eqref{eq:charges_8_Fracton}, this implies that
$\boldsymbol{\pi}_{1}=-\boldsymbol{\pi}_{2}$. In turn, this implies
that the generic fractonic dipole~\eqref{eq:Icangeneral-2-1-1} reduces
to \eqref{eq:L_fracton_8_reduced-1-2}, which is the action that
describes exclusively the degree of freedom that is associated with
the dipole moment.

\subsubsection{Neutral carrollian particle from dipoles}
\label{sec:derivation-case-6}

In the case when $q=0$ and $\db=\zerob$, from
\eqref{eq:charges_8_Fracton} one finds
$\boldsymbol{x}=\boldsymbol{x}_{1}=\boldsymbol{x}_{2}$. Therefore,
both monopoles are located at exactly the same position. This should
not be confused with the definition of an elementary dipole, where the
separation between the monopoles approaches zero while the electric
charge tends to infinity, resulting in a finite dipole moment
$\db \neq \zerob$.  In this case, the electric charge of each monopole
is opposite but finite, resulting in an elementary particle that is
electrically neutral and described by the following lagrangian
\begin{equation}
  \label{eq:L_fracton_8_reduced-2}
  L[\x,\pb,\eta_{1}]=\p\cdot\dot{\boldsymbol{x}}-E_{0}-\eta_{1}\left(\|\p\|^{2}-p^{2}\right).
\end{equation}
The equations of motion derived from this lagrangian are equivalent to
those of a massless carrollian particle with zero helicity which also
has vanishing total electric charge.

Fundamentally this relation is rooted in the fact that the
distinguishing feature, centre-of-mass and dipole charge, vanish. The
Carroll and fracton particles are not just dual to each other, but in
this case even physically equivalent and given by an aristotelion as
already discussed in Section~\ref{sec:massless-parallel}.

\section{Field theories and generalisation to curved space}
\label{sec:field-theor-gener}

In this section we provide some remarks concerning the relation to
known ``electric'' and ``magnetic'' carrollian field theories, the
difference between Carroll boost and dipole symmetry and comment on
the generalisation of the Carroll/fracton correspondence to curved space.

\subsection{Massive and massless Carroll field theories}
\label{sec:electr-magn-carr}

Given our understanding of the elementary particles it is natural to
ask if there are field theory realisations for which they can be seen
as excitations. We will elaborate this point in detail in our future
work~\cite{Figueroa-OFarrill:2023qty}, but let us nevertheless provide some remarks. To
that end we restrict for simplicity to massive spin zero and massless
carrollions with zero helicity.

Using Dirac quantisation, see, e.g.,~\cite[§13]{Henneaux:1992ig}, for
the constraint $E-E_{0}=0$ and using $E \mapsto i \pd_{t}$ we obtain
the following equation for a massive carrollian spin $0$ field
$\phi(t,\xb)$
\begin{align}
(i \pd_{t}-E_{0})\phi = 0 \, .
\end{align}
This equation can be derived from an action of the form
\begin{align}
  \label{eq:field-massive}
  S_{E_{0}\neq 0}[\phi,\phi^{*}]= \int dt d^{3}x \left( i \phi^{*}\dot  \phi -E_{0} \phi^{*} \phi \right) \, .
\end{align}
If we want to consider particles and antiparticles at once (in which
case we can for simplicity restrict to a real scalar) we obtain
\begin{align}
( \pd_{t}^{2} + E_{0}^{2})\phi = 0 \, ,
\end{align}
which we can derive from
\begin{align}
  \label{eq:field-massive-both}
  S_{E_{0}^{2}\neq 0}[\phi]=\frac{1}{2}\int dt d^{3}x \left( \dot \phi^{2} -E_{0}^{2} \phi^{2} \right) \, .
\end{align}
This action agrees with the ultralocal or ``electric'' field theories
considered
in~\cite{Klauder:1970cs,Bagchi:2019xfx,Henneaux:2021yzg,deBoer:2021jej}.
Similarly to the particles, see Section~\ref{sec:infin-symm-mass},
these actions also admit a symmetry enhancement and are not only
invariant under linear Carroll boosts but under
$\delta\phi = - f(\x)\pd_{t}\phi$ where $f(\x)$ is a free function
(these symmetries were called spacetime subsymmetries
in~\cite{Baig:2023yaz}).

For the massless carrollion with helicity $h=0$ we obtain using the
constraints $E=0$ and $\|\bm{p} \|^{2}-p^{2}=0$ the following equations
\begin{align}
  \label{eq:EOM-massless}
  \pd_{t}\phi&=0 & (\pd_{i}\pd^{i} + p^{2}) \phi=0 \, .
\end{align}
Let us contrast these equations with what is sometimes called
``magnetic'' carrollian theory in the
literature~\cite{Henneaux:2021yzg,deBoer:2021jej} (similar actions
have also appeared in the context of flat space holography and
deformations in lower
dimensions~\cite{Barnich:2012aw,Rodriguez:2021tcz})
\begin{align}
  S_{\mathrm{magnetic}}[\phi,\pi]
  = \int d t d^{3}x 
  \left(
  \pi\dot \phi - \frac{1}{2}\pd_{i}\phi \pd^{i} \phi + \frac{1}{2} p^{2}\phi^{2}
  \right)  \, .
\end{align}
The variation of $\pi$ indeed leads to the first equation
in~\eqref{eq:EOM-massless}, the variation of $\phi$ however provides
\begin{align}
  (\pd_{i}\pd^{i} + p^{2}) \phi=\dot \pi \, ,
\end{align}
which has a source term with respect to~\eqref{eq:EOM-massless}. One
could try to remedy this by considering an action of the form
\begin{align}
  S_{E_{0}=0}[\phi,\pi,u]
  = \int d t d^{3}x 
  \left(
  \pi\dot \phi - u (\Delta + p^{2})\phi
  \right)  \, .
\end{align}
which leads indeed to the equations~\eqref{eq:EOM-massless}, we leave
however further explorations to our future work~\cite{Figueroa-OFarrill:2023qty}.

\subsection{Carroll boost versus dipole symmetry for field theories}
\label{sec:dipole-vers-carr}

Let us provide some cautionary remarks concerning the Carroll/fracton
correspondence, for simplicity we restrict to scalar fields.

Following the action of Carroll symmetry on spacetime
\eqref{eq:carroll-trans} the Carroll boosts act as
\begin{align}
  \phi(t,\x) \mapsto \phi(t-\vb_{C}\cdot \x,\x)
\end{align}
on complex scalar fields, while linear dipole transformations act
as~\cite{Pretko:2018jbi}
\begin{align}
  \phi(t,\x) \mapsto e^{i \vb_{F}\cdot \x}\phi(t,\x) \, .
\end{align}
In this case it is clear that these symmetries are inequivalent, while
Carroll boosts are spacetime symmetries the dipole symmetries are
internal symmetries and do not act on the geometry.

There also exist theories that admit, both, either or none of these
symmetries, e.g., let us consider the complex scalar $\phi(t,\xb)$
theory~\cite{Pretko:2018jbi}
\begin{align}
  \label{eq:frac-field}
 S[\phi] = \int dt d^{3}x 
  \left(
  \dot \phi \dot \phi^{*} - \lambda X_{ij}X^{*}_{ij}
  \right)
\end{align}
where $X_{ij}= \pd_{i} \phi \pd_{j} \phi - \phi \pd_{i}\pd_{j}\phi$.
For $\lambda=0$ this theory has Carroll boost symmetry
$\delta\phi = -(\vb_{C} \cdot \x)\pd_{t}\phi$ and dipole symmetry
$\delta\phi = i (\vb_{F} \cdot \x)\phi$.\footnote{For $\lambda=0$ this
  theory has even more ``supertranslation-like'' symmetries (see,
  e.g., \cite[Section 2.4 and 2.6]{Bidussi:2021nmp}), but this is not
  important for the argument.} When $\lambda \neq 0$ only the dipole
symmetry remains and the theory has no Carroll boost symmetry. On the
other hand the real scalar field $L[\phi] = \frac{1}{2}\dot \phi^{2}$
has only Carroll boost symmetry and when the gradient term
$\frac{1}{2}\pd_{i}\phi\pd^{i}\phi$ is added it has neither.

From this perspective dipole symmetry is an internal symmetry, similar
to, e.g., internal spin degrees of freedom or internal $SU(n)$
symmetries. In particular for models of the type~\eqref{eq:frac-field}
dipole conservation is not related to a spacetime symmetry and is
therefore different from Carroll boosts~\cite{Bidussi:2021nmp}. For
this reason it was argued in~\cite{Bidussi:2021nmp,Jain:2021ibh} that
the geometry to which fracton theories of the type~\eqref{eq:frac-field} are coupled
is aristotelian, and therefore it does not admit boosts. See
also~\cite{Perez:2022kax} for the (asymptotic) analysis of the gauge
theory sector which also finds aristotelian symmetries.

What gives rise to the correspondence on the level of the particle is the
fact that coadjoint orbits and therefore the intrinsic definition of
elementary particles are based on the group structure and not
necessarily the underlying spacetime geometry. However it would be
interesting to see if there is more to be learned about this correspondence
between internal and external symmetries, see, e.g., the interesting
recent works~\cite{Baig:2023yaz,Kasikci:2023tvs,Huang:2023zhp}.

\subsection{(A)dS Carroll and fractons on curved space }
\label{sec:ads-carroll-fractons}

Our discussions so far were focused on flat Carroll space, but we
would like to mention a possible generalisation to curved space. More
precisely to (A)dS Carroll which can be thought of the carrollian
analogs of (anti-)de Sitter space~\cite{Figueroa-OFarrill:2018ilb}
from which they arise as a limit~\cite{Bacry:1968zf}.

On the level of the symmetries this means that the flat Carroll
symmetries~\eqref{eq:carroll-kla-brackets} have the following
additional commutation relations
\begin{align}
  \label{eq:AdS-carroll}
  [P_{a},P_{b}]&=-\Lambda J_{ab} & [P_{a},H] = \Lambda B_{a}
\end{align}
where the cosmological constant $\Lambda < 0$ leads to AdS Carroll and
$\Lambda > 0$ to dS Carroll. They share similarities with their
lorentzian counter parts and are therefore interesting candidates for
holography and cosmology. But rather than lorentzian they also have
carrollian boosts. Another interesting property of AdS Carroll is that
upon the exchange of boosts and translations the symmetries are
isomorphic to Poincaré symmetries. This means that for this case the
particles should be describable in terms of well known Poincaré
particles. Furthermore, AdS Carroll is closely related to time-like
infinity of asymptotically flat
spacetimes~\cite{Figueroa-OFarrill:2021sxz}.

Using our correspondence means that the dipole algebra~\eqref{eq:frac-alg}
obtains the following additional commutation relations
\begin{align}
  \label{eq:AdS-fracton}
  [P_{a},P_{b}]&=-\Lambda J_{ab} & [P_{a},Q] = \Lambda D_{a} \, .
\end{align}
The first commutation relation implies that they are now living on
hyperbolic space ($\Lambda <0$) or on a sphere ($\Lambda > 0$). The
second commutation relation implies that the charge is no longer
central and that it is related to the dipole moment via the Casimir
$\Lambda Q^{2} +\| \bm{D} \|$. For $\Lambda <0$ we see an emergent
Lorentz symmetry, of course related to the underlying Poincaré
symmetry. It would be interesting to further explore the
Carroll/fracton correspondence in these curved spaces.

\section{Discussion and outlook}
\label{sec:discussion-outlook}

This work provides a definition and classification of classical
Carroll particles and fractons in $3+1$ dimensions, summarised in
Table~\ref{tab:carrollions}. Based on the known relation between
Carroll and dipole symmetries and their free
theories~\cite{Bidussi:2021nmp} (see also~\cite{Marsot:2022imf}) we
propose a correspondence on the level of the elementary particles, which is
summarised in Table~\ref{tab:carrvsfrac}, and show that while their
physical interpretations differ they are indeed equivalent (at least
classically on the reduced phase space).

The Carroll/fracton correspondence is indeed useful to obtain physical
insights. For instance, isolated massive carrollions are stuck to a
point due to the conservation of the center-of-mass charge for the
very same reason that fracton monopoles are stuck to a point due to dipole
conservation.  A carrollian way to think about this property would be
to think about the closing of the Minkowskian light cone which also
implies immobility in space. On the other hand it might be useful to
think about massless carrollions as moving dipoles.

Given that both of these subjects connect to various interesting areas
of current research it is clear that many things can be said. Let us
now relate our results to various other interesting topics and provide
some areas for future exploration.
\begin{description}[style=nextline]
\item[Quantum Carroll/fracton particles] It is well-known that
  coadjoint
  orbits~\cite{MR1461545,10.1007/BFb0079068,Kirillov2004Lectures} and
  particles actions~\cite{Alekseev:1988vx} provide a fruitful starting
  point for quantisation. In a future work~\cite{Figueroa-OFarrill:2023qty} we will look
  at the quantum particles to which this correspondence generalises.

\item[Field theories] It is natural to try to systematically connect
  these particles to field theories on Carroll spacetime. We will
  show~\cite{Figueroa-OFarrill:2023qty} that of the unitary irreducible representations of
  the Carroll group, there are two classes of representations which
  can be realised as (finite-component) fields on Carroll spacetime.
  The first class are the massive carrollions (except that the spin is
  quantised) and the second class as the massless carrollions with
  helicity (which is also quantised). The former are related to
  electric field theories, whereas the latter to magnetic field
  theories, as discussed in Section~\ref{sec:electr-magn-carr}.

\item[Relation to time-like symmetries] In~\cite{Gorantla:2022eem}
  fractons were described using ``time-like'' higher-form global
  symmetries. It might be interesting to understand the relation
  between these generalised symmetries and our results.

\item[Planons, lineons and other exotic particles] This work
  highlights the applicability of the orbit method beyond the
  conventional framework and we will show that this also generalises
  to other exotic particles with restricted mobility~\cite{IP:2023}.
  
  In particular the symmetries of planons are isomorphic to the
  Bargmann (=centrally extended Galilei)
  algebra~\cite{Gromov:2018nbv}. Using the methods described in this
  work we can relate the respective particles, e.g., planons are
  related to massless galilean particles~\cite{IP:2023}.

  The worldline description has also been applied to fractonic
  theories with subsystem symmetries~\cite{Casalbuoni:2021fel}.

\item[Other dimensions] Much of what has been said should be
  generalisable to generic dimension. Let us however remark that in
  $2+1$ dimensions there are nontrivial central extensions which would
  make this correspondence more involved, cf.,~\cite{Marsot:2021tvq}.

\item[Fractons in flat holography and black holes] Given that
  carrollian symmetries have emerged in flat holography~(see, e.g.,
  \cite{Duval:2014uva,Figueroa-OFarrill:2021sxz,Donnay:2022aba,Bagchi:2022emh,Bekaert:2022oeh,Saha:2023hsl,Salzer:2023jqv})
  and for black holes (see, e.g.,
  \cite{Donnay:2019jiz,Marsot:2022imf}) it is intriguing to try
  understand them from a fractonic perspective, see also our comments
  in Section~\ref{sec:ads-carroll-fractons}.
\end{description}

\acknowledgments

We thank Glenn Barnich, Andrea Campoleoni, Jelle Hartong, Emil Have, Simon Pekar and Ali Seraj 
for useful discussions. The research of AP is partially supported by Fondecyt grants No
1211226, 1220910 and 1230853. SP is supported by the Leverhulme Trust Research
Project Grant (RPG-2019-218) ``What is Non-Relativistic Quantum
Gravity and is it Holographic?''.

\appendix

\section{Carroll symmetry}
\label{app:carroll-symmetry}

In this appendix we define the Caroll algebra, the Carroll group,
discuss the adjoint and coadjoint actions, automorphisms and their
effect on coadjoint orbits and the Maurer--Cartan one-forms.  We do
this in general dimension before specialising to dimension $3+1$.

\subsection{The Carroll group and its Lie algebra}
\label{sec:carroll-group}

The Carroll group acts via affine transformations on Carroll
spacetime, the ultra-relativistic limit of Minkowski spacetime
\cite{MR0192900,SenGupta1966OnAA}. The $(n+1)$-dimensional Carroll
algebra $\g$, by which we mean the Carroll algebra acting on
$(n+1)$-dimensional Carroll spacetime, is spanned by
$J_{ab}=-J_{ba}, B_a, P_a, H$, with $a,b=1,\dots,n$, with nonzero Lie
brackets
\begin{equation}
  \label{eq:carroll-kla-brackets}
  \begin{split}
    [J_{ab},J_{cd}] &= \delta_{bc} J_{ad} - \delta_{ac} J_{bd} -  \delta_{bd} J_{ac} + \delta_{bd} J_{ac} \\
    [J_{ab}, B_c] &= \delta_{bc} B_a - \delta_{ac} B_b\\
    [J_{ab}, P_c] &= \delta_{bc} P_a - \delta_{ac} P_b\\
    [B_a, P_b] &= \delta_{ab} H.
  \end{split}
\end{equation}
We may embed the Carroll Lie algebra $\g$ in $\gl(n+2,\RR)$ as
follows:
\begin{equation}
  \label{eq:affine-embedding}
  \tfrac12 X^{ab} J_{ab} + v^a B_a + a^a P_a + s H \mapsto
  \begin{pmatrix}
    X & \bzero & \ba \\
    \bv^T & 0 & s\\
    \bzero^T & 0 & 0
  \end{pmatrix},
\end{equation}
where $X^T = - X \in \so(n)$. We may parametrise the (connected)
Carroll group as follows:
\begin{equation}
\label{eq:grpara}
  g(R,\bv,\ba,s) = \exp(sH) \exp(v^aB_a + a^a P_a) R \, ,
\end{equation}
where $R \in \SO(n)$.  This non-standard parametrisation of the
Carroll group has the advantage that it puts $B_{a}$ and $P_{a}$ on
equal footing, reflecting the fact that they can be mapped into each
other under automorphisms, as we discuss in
Section~\ref{sec:automorphisms}. This leads to more symmetric
equations that make this symmetry manifest, at the cost that some
equations are more complicated.  We discuss in the next
subsection~\ref{sec:symm-brok-param} parametrisations which are more
economical for other aspects, e.g., when acting on the spacetime.

The resulting group is seen to be the subgroup of $\GL(n+2,\RR)$
consisting of matrices of the form
\begin{equation}
  \label{eq:carroll-group-element}
  \begin{pmatrix}
    R & \bzero & \ba \\
    \bv^TR & 1 & s + \tfrac12 \bv^T\ba\\
    \bzero^T & 0 & 1
  \end{pmatrix},
\end{equation}
with $R \in \SO(n)$, from where we can work out the group
multiplication
\begin{equation}
  \label{eq:group-multiplication}
    \begin{pmatrix}
    R_1 & \bzero & \ba_1 \\
    \bv_1^TR_1 & 1 & s_1 + \tfrac12 \bv_1^T\ba_1\\
    \bzero^T & 0 & 1
  \end{pmatrix}
  \begin{pmatrix}
    R_2 & \bzero & \ba_2 \\
    \bv_2^TR_2 & 1 & s_2 + \tfrac12 \bv_2^T\ba_2\\
    \bzero^T & 0 & 1
  \end{pmatrix}
  =
  \begin{pmatrix}
    R_3 & \bzero & \ba_3 \\
    \bv_3^TR_3 & 1 & s_3 + \tfrac12 \bv_3^T\ba_3\\
    \bzero^T & 0 & 1
  \end{pmatrix}
\end{equation}
where
\begin{equation}
  \begin{split}
    R_3 &= R_1 R_2\\
    \bv_3 &= \bv_1 + R_1 \bv_2\\
    \ba_3 &= \ba_1 + R_1 \ba_2\\
    s_3 &= s_1 + s_2 + \tfrac12 \bv_1^T R_1 \ba_2 - \tfrac12 \ba_1^T R_1 \bv_2.
  \end{split}
\end{equation}
It is then straightforward to work out the inverse of the generic element:
\begin{equation}
  \label{eq:inverse}
    \begin{pmatrix}
    R & \bzero & \ba \\
    \bv^TR & 1 & s + \tfrac12 \bv^T\ba\\
    \bzero^T & 0 & 1
  \end{pmatrix}^{-1} =
  \begin{pmatrix}
    R^T & \bzero & - R^T \ba\\
    -\bv^T & 1 & -s + \tfrac12 \ba^T R \bv\\
    \bzero^T & 0 & 1
  \end{pmatrix},
\end{equation}
where we have used that $R^T R = \mathbb{1}$.

Finally, identifying Carroll spacetime with the affine hyperplane of
$\RR^{n+2}$ consisting of those vectors whose last entry is equal to
$1$, we work out the action of the Carroll group on Carroll spacetime:
\begin{equation}
  \label{eq:spacetimeaction}
    \begin{pmatrix}
    R & \bzero & \ba \\
    \bv^TR & 1 & s + \tfrac12 \bv^T\ba\\
    \bzero^T & 0 & 1
  \end{pmatrix}
  \begin{pmatrix}
    \x \\ t \\ 1
  \end{pmatrix} =
  \begin{pmatrix}
    R \x + \ba \\ t + s + \tfrac12 \bv^T\ba + \bv^T R \x \\ 1
  \end{pmatrix}
\end{equation}
from where we see that the action (in this non-standard
parametrisation) of the Carroll group consists of a rotation
$(\x,t) \mapsto (R \x,t)$, followed by a Carroll boost
$(R\x,t) \mapsto (R\x, t + \bv^T R \x)$ followed in turn by a
translation
$(R\x, t + \bv^T R\x) \mapsto (R\x + \ba, t + \bv^T R\x + s + \tfrac12
\bv^T \ba)$. In the next subsection we will discuss a parametrisation
of the group element that is more economical.

\subsubsection{Non-symmetric group parametrisations}
\label{sec:symm-brok-param}

When we do not insist on a parametrisation that puts the two vectors
$B_{a}$ and $P_{a}$ on equal footing there are other useful choices,
in particular when we are interested in the group action on the
spacetime. As a first step we parametrise the group element as
\begin{align}
  \label{eq:carroll-group-par-broken}
  g(R,\bv,\ba,s') = \exp(s'H) \exp(a^a P_a) \exp(v^a B_a) R \, .
\end{align}
Using
$e^{a^{a} P_{a} + v^{a}B_{a}} = e^{\frac{1}{2} \bv \cdot \ba H}
e^{a^{a} P_{a}} e^{v^{a} B_{a}}$ we can relate this parametrisation to
the one above~\eqref{eq:grpara} by $s'= s + \tfrac12 \bv\cdot\ba $
which leads to the matrix representation
\begin{equation}
  \label{eq:carroll-group-element-broken}
  \begin{pmatrix}
    R& \bzero & \ba \\
    \bv^TR & 1 & s' \\
    \bzero^T & 0 & 1
  \end{pmatrix} \, .
\end{equation}

When our main concern is the action on the underlying spacetime the
following parametrisation is particularly useful
\begin{align}
  \label{eq:carroll-group-para-spacetime}
  g( R,\bv',\ba,s') = \exp(s'H) \exp(a^a P_a) R \exp(v'^a B_a)  \, .
\end{align}
Using
$R^{-1} e^{\bm{v} \cdot \bm{B}} R = e^{(R^{-1}\vb) \cdot \bm{B}}$ it
can be related to \eqref{eq:carroll-group-par-broken}
via~$\bv= R^{-1}\bv'$ and consequently we can write the group element
as
\begin{equation}
  \label{eq:carroll-group-element-space}
  \begin{pmatrix}
    R & \bzero & \ba \\
    \vb'^{T} & 1 & s' \\
    \bzero^T & 0 & 1
  \end{pmatrix} \, .
\end{equation}
The group law may look unconventional
\begin{equation}
  \begin{split}
    R_3 &= R_1 R_2\\
    \bv'_3 &= R_{2}^{-1}\bv'_1 + \bv'_2\\
    \ba_3 &= \ba_1 + R_1 \ba_2\\
    s'_3 &= s'_1 + s'_2 + \bv'_{1} \cdot \ab_{2} \, ,
  \end{split}
\end{equation}
but acting with the group
element~\eqref{eq:carroll-group-element-space} on the spacetime, as
in~\eqref{eq:spacetimeaction}, leads precisely to the simple
transformation of the spacetime given in~\eqref{eq:carroll-trans}
(where we have dropped the primes).

\subsection{Automorphisms}
\label{sec:automorphisms}

The group of automorphisms of the Carroll Lie algebra (for any $n\geq 3$) which
fix the rotational subalgebra is isomorphic to $\GL(2,\RR)$, with
$\begin{pmatrix}a & b \\ c & d\end{pmatrix} \in \GL(2,\RR)$ acting as
\begin{equation}
  \label{eq:autos-carroll}
  J_{ab} \mapsto J_{ab}, \quad B_a \mapsto a B_a + b P_a, \quad P_a \mapsto
  c B_a + d P_a \quad\text{and}\quad H \mapsto (ad - bc) H.
\end{equation}
These automorphisms are not inner: they do not arise by conjugation in
the Carroll group.  These automorphisms act on the dual $\g^*$ of the
Lie algebra as follows.  If we let $\lambda^{ab}, \beta^a, \pi^a,
\eta$ denote the canonical dual basis to $J_{ab}, B_a, P_a, H$, we see
that
\begin{equation}
  \lambda^{ab} \mapsto \lambda^{ab}, \quad \beta^a \mapsto
  \tfrac1{ad-bc} (d \beta^a - c \pi^a),\quad \pi^a \mapsto
  \tfrac1{ad-bc}  (-b \beta^a + a \pi^a) \quad\text{and}\quad \eta
  \mapsto \tfrac1{ad-bc}  \eta.
\end{equation}
We may use these automorphisms to relate coadjoint orbits which
otherwise might seem unrelated. In Appendix~\ref{sec:acti-autom-coadj}
we show how group automorphisms act on coadjoint orbits and then in
Appendix~\ref{sec:orbits-mod-autos} we will see how this allows us to
simplify the classification of coadjoint orbits for $n=3$.

\subsection{The adjoint and coadjoint actions}
\label{sec:adjo-coadj-acti}

We now work out the adjoint and coadjoint actions of the (connected)
Carroll group on its Lie algebra and its dual.  Having embedded the
Carroll group inside $\GL(n+2,\RR)$, the adjoint action is simply
conjugation.

Consider the following matrix $A \in \g$,
\begin{equation}
  \label{eq:lie-algebra-matrix}
  A = \begin{pmatrix}
    X & \bzero & \ba\\
    \bb^T & 0 & c\\
    \bzero^T & 0 & 0
  \end{pmatrix}
\end{equation}
and let us conjugate by a generic group element $g \in G$, given by
\eqref{eq:carroll-group-element}, whose inverse is given by
\eqref{eq:inverse}.  We obtain
\begin{equation}
  \label{eq:adjoint-action}
  \Ad_g  A =
  \begin{pmatrix}
    R & \bzero & \ba \\
    \bv^TR & 1 & s + \tfrac12 \bv^T\ba\\
    \bzero^T & 0 & 1
  \end{pmatrix}
  \begin{pmatrix}
    X & \bzero & \ba\\
    \bb^T & 0 & c\\
    \bzero^T & 0 & 0
  \end{pmatrix}
  \begin{pmatrix}
    R^T & \bzero & - R^T \ba\\
    -\bv^T & 1 & -s + \tfrac12 \ba^T R \bv\\
    \bzero^T & 0 & 1
  \end{pmatrix}=
  \begin{pmatrix}
    X' & \bzero & \ba'\\
    \bb'^T & 0 & c'\\
    \bzero^T & 0 & 0
  \end{pmatrix},
\end{equation}
where
\begin{equation}
  \label{eq:adjoint-action-too}
  \begin{split}
    X' &= R X R^T\\
    \ba' &= R \ba - R X R^T \ba\\
    \bb' &= R \bb - R X R^T \bv\\
    c' &=  c + \bv^T R \ba - \ba^T R \bb - \bv^T R X R^T \ba.
  \end{split}
\end{equation}

We now define an inner product on $\g$ by
\begin{equation}
  \left<
    \begin{pmatrix}
      X_1 & \bzero & \ba_1\\ \bb_1^T & 0 & c_1 \\ \bzero^T & 0 & 0
    \end{pmatrix}
    ,
    \begin{pmatrix}
      X_2 & \bzero & \ba_2\\ \bb_2^T & 0 & c_2 \\ \bzero^T & 0 & 0
    \end{pmatrix}\right>
  = \tfrac12 \Tr\left( X_1^T X_2  \right) +
  \bb_1^T \bb_2 + \ba_1^T \ba_2 + c_1 c_2.
\end{equation}
In this way we may identify $\g^*$ as a vector space with $\g$ with
the dual pairing being the above inner product.  Let $\alpha \in
\g^*$, then
\begin{equation}
  \left<\Ad^*_g \alpha, A\right> = \left<\alpha, \Ad_{g^{-1}} A\right>.
\end{equation}
We may calculate $\Ad_{g^{-1}}A$ as follows
\begin{equation}
  \label{eq:ad-ginv}
  \begin{split}
    \Ad_{g^{-1}}
    \begin{pmatrix}
      X & \bzero & \ba\\
      \bb^T & 0 & c\\
      \bzero^T & 0 & 0
    \end{pmatrix}&= 
  \begin{pmatrix}
    R^T & \bzero & - R^T \ba\\
    -\bv^T & 1 & -s + \tfrac12 \ba^T R \bv\\
    \bzero^T & 0 & 1
  \end{pmatrix}
      \begin{pmatrix}
      X & \bzero & \ba\\
      \bb^T & 0 & c\\
      \bzero^T & 0 & 0
    \end{pmatrix}
  \begin{pmatrix}
    R & \bzero & \ba \\
    \bv^TR & 1 & s + \tfrac12 \bv^T\ba\\
    \bzero^T & 0 & 1
  \end{pmatrix}\\
  &=
  \begin{pmatrix}
    X' & \bzero & \ba'\\
    \bb'^T & 0 & c'\\
    \bzero^T & 0 & 0
  \end{pmatrix}~,
  \end{split}
\end{equation}
where
\begin{equation}
  \begin{split}
    X' &= R^T X R\\
    \ba' &= R^T(\ba + X \ba)\\
    \bb' &= R^T(\bb + X \bv)\\
    c' &= c + \ba^T \bb - \bv^T (\ba + X \ba).
  \end{split}
\end{equation}
Let $\alpha \in \g^*$ be given by
\begin{equation}\label{eq:generic-covector}
  \alpha =
  \begin{pmatrix}
    J & \bzero & \p\\ \k^T & 0 & E \\ \bzero^T & 0 & 0
  \end{pmatrix},
\end{equation}
with $J^T = - J$, and $g \in G$ the generic element in
\eqref{eq:carroll-group-element}.  Then we have that
\begin{equation}
  \label{eq:coadjoint-action}
  \Ad_g^* \alpha =
  \begin{pmatrix}
    J' & \bzero & \p'\\ \k'^T & 0 & E' \\ \bzero^T & 0 & 0
  \end{pmatrix},
\end{equation}
where
\begin{equation}
  \tfrac12 \Tr (J'^T X) + \bb^T \k' + \ba^T \p' + c E' =   \tfrac12 \Tr (J^T X') + \bb'^T \k + \ba'^T \p + c' E,
\end{equation}
from where we can read off
\begin{equation}
  \label{eq:coadjoint-rep}
  \begin{split}
    J' &= R J R^T +  (Rk)\bv^T - \bv (R\k)^T + R\p \ba^T -  \ba (R\p)^T + E \left( \ba \bv^T - \bv\ba^T  \right) \\
    \k' &= R\k + E \ba\\
    \p' &= R\p - E\bv\\
    E' &= E.
  \end{split}
\end{equation}
This coadjoint action was already discussed in~\cite[Appendix
A]{Duval:2014uoa}.

\subsection{Maurer--Cartan one-form and particle actions}
\label{sec:maurer-cartan-one}

Let $\alpha \in \g^*$ and let $\omega_{\text{KKS}}$ denote the
KKS invariant symplectic form on the coadjoint
orbit $\eO_\alpha$ of $\alpha$.  This coadjoint orbit is by definition
a homogeneous symplectic manifold of the Carroll group and choosing
$\alpha$ as the base point, we can define an orbit map $\pi: G \to
\eO_\alpha$ by $g \mapsto \Ad^*_g \alpha$.  Pulling back
$\omega_{\text{KKS}}$ via $\pi$ we see that it is not just closed, but
actually exact
\begin{equation}
  \pi^*\omega_{\text{KKS}} = - d \left<\alpha,\theta^L\right>,
\end{equation}
where $\theta^L$ is the left-invariant $\g$-valued Maurer--Cartan
one-form.  The one-form $\left<\alpha,\theta^L\right>$ is the main
ingredient in the construction of particle actions associated to the
coadjoint orbit and hence it is convenient to record here the one-form
relative to our chosen group parametrisation.

Parametrising the generic group element as in
\eqref{eq:carroll-group-element}, we find that the pull-back of the
left-invariant Maurer--Cartan one-form is given by
\begin{equation}
  \label{eq:LI-MC-one-form}
  \begin{split}
    g^{-1}dg &=
    \begin{pmatrix}
      R^T & \bzero & -R^T\ba\\ -\bv^T & 1 & -s + \tfrac12 \ba^T R\bv\\
      \bzero^T & 0 & 1
    \end{pmatrix}
    \begin{pmatrix}
      dR & \bzero & d\ba\\ d\bv^T R + \bv^T dR & 0 & ds + \tfrac12 d\bv^T
      v + \tfrac12 \bv^T d\ba\\ \bzero^T & 0 & 0
    \end{pmatrix}\\
    &=
    \begin{pmatrix}
      R^TdR & \bzero & R^T d\ba\\ d\bv^T R & 0 & ds + \tfrac12 \ba^T
      d\bv - \tfrac12 \bv^Td\ba\\ \bzero^T & 0 & 0
    \end{pmatrix}.
  \end{split}
\end{equation}
Therefore for $\alpha \in \g^*$ given by equation~\eqref{eq:generic-covector},
\begin{equation}
  \left<\alpha, g^{-1}dg\right> = \tfrac12 \Tr J^T R^T dR + (R\p -
  \tfrac12 E \bv)^T d\ba + (R\k + \tfrac12 E \ba)^T d\bv + Eds.
\end{equation}

The particle action associated with the coadjoint orbit of $\alpha \in
\g^*$ is defined as follows.  Let $I \subset \RR$ be a real interval
parametrising a curve $g : I \to G$ in the group.  We define the
action functional
\begin{equation}
  S[g] := \int_I \left<\alpha, g^*\theta^L\right> = \int_I
  \left<\alpha, g^{-1} \dot g \right> d\tau
\end{equation}
where $\tau$ denotes the coordinate on the interval.  Varying the
action, we obtain
\begin{align*}
  \delta S &= \int_I \left<\alpha, -g^{-1}\delta g g^{-1} \dot g + g^{-1}\delta\dot g\right> d\tau\\
           &= \int_I \left<\alpha, [g^{-1}\dot g,   g^{-1}\delta g]\right> d\tau + \int_{\partial I} \left<\alpha, g^{-1}\delta g\right> &\tag{integrating by parts}\\
           &= - \int_I \left<\ad_{g^{-1}\dot g}^*\alpha, g^{-1}\delta g\right> d\tau \,  ,
\end{align*}
where we have discarded the boundary term
$\int_{\partial I} \left<\alpha,g^{-1}\delta g\right>$,
assuming an endpoint-fixed variational problem.  The variation of the
action vanishes for all $\delta g$ if and only if
$g^{-1}\dot g$ takes values in the stabiliser subalgebra
$\g_\alpha$ of $\alpha$. One might be tempted to think that this
requires $g(\tau)$ to be in the stabiliser subgroup $G_\alpha$,
but recall that $g^{-1}\dot g$ is the pull-back of a
left-invariant one-form, so the solution is actually
$g(\tau) = g_0 h(\tau)$ for some $h :I \to G_\alpha$ and where
$g_0 \in G$ is a constant element of the group. Pushing down this curve
via the orbit map $\pi: G \to \eO_\alpha$, produces
$\Ad_{g_0}^*\alpha \in \eO_\alpha$. So the extremals of the action do not
necessarily have momenta $\alpha$: all we can say is that their
momenta lie in the coadjoint orbit of $\alpha$.  This merely
highlights the fact that the action is indeed associated with the
coadjoint orbit and not with any orbit representative $\alpha$.

In Souriau's language, but going back to Lagrange, the coadjoint orbit
is the space of motions: a point in the coadjoint orbit represents a
trajectory.  Particle actions whose extremals are curves live in an
evolution space fibering over the coadjoint orbit.  The difficulty in
describing the evolution space intrinsically can be circumvented by
lifting the trajectories to the group as we have done above; even
though doing so, as mentioned in the bulk of the paper, results
in a redundancy in the description; in effect, in gauge invariance.

\subsection{The case $n=3$}
\label{sec:case-n=3}

In $n=3$ the adjoint and vector representations of $\so(3)$ are
isomorphic.  We can therefore trade $3\times 3$ skewsymmetric matrices
for vectors.  Let us define the linear map $\varepsilon : \RR^3 \to
\so(3)$ by
\begin{equation}
  \varepsilon(\ba) \bb = \ba \times \bb.
\end{equation}
This belongs to $\so(3)$ because $\varepsilon(\ba)\bb \cdot \bb = 0$,
so the endomorphism $\varepsilon(\ba)$ is skew-symmetric.  Explicitly,
\begin{equation}
  \label{eq:epsilon-rep}
  \varepsilon(\be_1) =   \begin{pmatrix}
    \zero & \zero & \zero \\
    \zero & \zero & -1 \\
    \zero & 1 & \zero 
  \end{pmatrix}\qquad\qquad
  \varepsilon(\be_2) =  \begin{pmatrix}
    \zero & \zero & 1 \\
    \zero & \zero & \zero \\
    -1 & \zero & \zero 
  \end{pmatrix}\qquad\qquad
  \varepsilon(\be_3) = \begin{pmatrix}
    \zero & -1 & \zero \\
    1 & \zero & \zero \\
    \zero & \zero & \zero 
  \end{pmatrix},
\end{equation}
so that $\varepsilon(\be_i)_{jk} = - \epsilon_{ijk}$. It follows from
the standard vector identities that $\varepsilon$ is a Lie algebra
isomorphism provided that we use the cross product to define the Lie
algebra structure on $\RR^3$:
\begin{equation}
  [\varepsilon(\ba), \varepsilon(\bb)] = \varepsilon(\ba \times \bb).
\end{equation}
Moreover, a standard calculation shows that $\varepsilon$ is an
isometry provided that we use the standard euclidean inner product on
$\RR^3$ and half the trace in the defining representation on $\so(3)$:
\begin{equation}
  \tfrac12 \Tr \varepsilon(\ba)^T \varepsilon(\bb) = \ba \cdot \bb.
\end{equation}
If $R \in \SO(3)$, it follows that
\begin{equation}
  R(\ba \times \bb) = (R\ba) \times (R\bb),
\end{equation}
which implies that
\begin{equation}
  \varepsilon( R \ba) = R \varepsilon(\ba) R^T.
\end{equation}
Finally, it follows from a straightforward calculation that
\begin{equation}
  \ba \bb^T - \bb \ba^T = \varepsilon(\bb \times \ba).
\end{equation}
Taking these formula into account and letting $J = \varepsilon(\bj)$,
we may write the coadjoint action of the group element $g \in G$ given
in \eqref{eq:carroll-group-element} on
\begin{equation}
  \alpha = \begin{pmatrix}
    \varepsilon(\bj) & \bzero & \p\\ \k^T & 0 & E \\ \bzero^T & 0 & 0
  \end{pmatrix} \in \g^*
\end{equation}
as
\begin{equation}
  \label{eq:coadjoint-action-n=3}
    \Ad_g^* \alpha =   \begin{pmatrix}
    \varepsilon(\bj') & \bzero & \p'\\
    \k'^T & 0 & E'\\
    \bzero^T & 0 & 0
  \end{pmatrix},
\end{equation}
where
\begin{equation}
  \label{eq:coadjoint-rep-n=3}
  \begin{split}
    \bj' &= R \bj + \bv \times R\k + \ba \times R\p + E\bv \times \ba\\
    \k' &= R\k + E \ba\\
    \p' &= R\p - E\bv\\
    E' &= E.
  \end{split}
\end{equation}

\subsection{Action of automorphisms on coadjoint orbits}
\label{sec:acti-autom-coadj}

In this section we will show that group automorphisms map coadjoint
orbits to coadjoint orbits symplectomorphically.  Later in
Appendix~\ref{sec:orbits-mod-autos} we apply this to further simplify
the classification of coadjoint orbits of the Carroll group.

Let $G$ be a Lie group and $\tau : G \to G$ an automorphism; that is,
a diffeomorphism which is also a group homomorphism $\tau(e) = e$ and
$\tau(ab) = \tau(a)\tau(b)$ for all $a,b \in G$.  Differentiating at
the identity we get $\tau_* : \g \to \g$, which is an automorphism of
the Lie algebra and moreover $\tau(\exp X) = \exp \tau_*X$ for all $X
\in \g$.  The invertible linear transformation $\tau_* \in \GL(\g)$
induces an invertible linear transformation $\tau^* \in \GL(\g^*)$ by
$\tau^*\alpha = \alpha \circ \tau_*^{-1}$ for $\alpha \in \g^*$.  A
natural question is how the coadjoint orbits of $\alpha$ and
$\tau^*\alpha$ are related.

\begin{lemma}
  Let $\eO_\alpha$ denote the coadjoint orbit of $\alpha \in \g^*$.
  Then $\eO_{\tau^*\alpha} = \tau^* \eO_\alpha$.
\end{lemma}

\begin{proof}
  Let $g \in G$, $\alpha \in \g^*$ and $X \in \g$.  Then
  \begin{align*}
    \left<\Ad^*_g \tau^*\alpha, X\right> &= \left<\tau^*\alpha, \Ad_{g^{-1}} X\right>\\
                                         &= \left<\alpha, \tau_*^{-1} \Ad_{g^{-1}} X\right>.
  \end{align*}
  But now
  \begin{align*}
    \tau_*^{-1} \Ad_{g^{-1}} X &= \left.\frac{d}{dt} \tau^{-1} (\exp(t \Ad_{g^{-1}}X ))\right|_{t=0} &\tag{$\tau_*^{-1} = (\tau^{-1})_*$}\\
                               &= \left.\frac{d}{dt} \tau^{-1} (g^{-1} \exp(t X ) g)\right|_{t=0}\\
                               &= \left.\frac{d}{dt} \left( \tau^{-1} (g^{-1}) \tau^{-1}(\exp(t X )) \tau^{-1}(g)\right)\right|_{t=0} &\tag{$\tau^{-1}$ is an automorphism of $G$}\\
                               &=\Ad_{\tau^{-1}(g^{-1})} \tau^{-1}_* X.
  \end{align*}
  Therefore,
  \begin{align*}
    \left<\Ad^*_g \tau^*\alpha, X\right> &= \left<\alpha, \Ad_{\tau^{-1}(g^{-1})} \tau^{-1}_* X\right>\\
                                         &= \left<\Ad^*_{\tau^{-1}(g)} \alpha, \tau^{-1}_* X\right> &\tag{$\tau^{-1}(g^{-1}) = \tau^{-1}(g)^{-1}$}\\
                                         &= \left<\tau^*\Ad^*_{\tau^{-1}(g)} \alpha, X\right>,
  \end{align*}
  so that
  \begin{equation}
    \label{eq:coad-tau}
    \Ad^*_g \tau^*\alpha =\tau^*\Ad^*_{\tau^{-1}(g)} \alpha.
  \end{equation}
  This finally implies that
  \begin{align*}
    \eO_{\tau^*\alpha} &= \left\{ \Ad^*_g \tau^*\alpha \middle | g \in G\right\} \\
                       &= \left\{ \tau^*\Ad^*_{\tau^{-1}(g)} \alpha \middle | g \in G\right\} \\
                       &= \tau^* \left\{\Ad^*_{\tau^{-1}(g)} \alpha \middle | g \in G\right\} \\
                       &= \tau^* \eO_\alpha,
  \end{align*}
  since $g = \tau^{-1}(\tau(g))$ for all $g \in G$.
\end{proof}

Of course, if $\tau$ is an inner automorphism, so that $\tau(g) = h g
h^{-1}$ for some $h \in G$, it follows that $\tau^*\alpha = \Ad^*_h
\alpha$ and hence $\eO_{\tau^*\alpha} = \eO_{\alpha}$.  Hence only
outer automorphisms relate different coadjoint orbits.

More is true and the diffeomorphism $\tau^*$ of $\g^*$ relates the
KKS symplectic forms on $\eO_\alpha$ and on $\eO_{\tau^*\alpha}$.

\begin{lemma}
  Let $\omega \in \Omega^2(\eO_{\tau^*\alpha})$ denote the KKS
  symplectic form on $\eO_{\tau^*\alpha}$.  Then $(\tau^*)^*\omega$
  agrees with the KKS symplectic form on $\eO_\alpha$.
\end{lemma}

\begin{proof}
  Let us introduce the notation $\phi := \tau^*$ to denote the
  diffeomorphism of $\g^*$ given by the action of the automorphism
  $\tau$.  (This declutters the notation somewhat.)  Then we wish to
  show that $\phi^*\omega$ is the KKS symplectic form on $\eO_\alpha$.
  Let $X \in \g$ and let $\left.\ad^*_X\right|_\alpha \in T_\alpha \g^*$
  be defined by
  \begin{equation}
   \left.\ad^*_X\right|_\alpha :=\left.\frac{d}{dt} \Ad^*_{\exp(t X)}
      \alpha \right|_{t=0}.
  \end{equation}
  Then the KKS symplectic structure $\omega_{\mathrm{KKS}}$ on
  $\eO_\alpha$ is defined by at $\alpha \in \g^*$ by
  \begin{equation}
    \omega_{\mathrm{KKS}} (\left.\ad^*_X\right|_\alpha, \left.\ad^*_Y\right|_\alpha ) = \left<\alpha, [X,Y]\right>.
  \end{equation}
  On the other hand,
  \begin{equation}
    (\phi^*\omega)_\alpha (\left.\ad^*_X\right|_\alpha,
    \left.\ad^*_Y\right|_\alpha )= \omega_{\phi(\alpha)}(\phi_*
    \left.\ad^*_X\right|_\alpha, \phi_* \left.\ad^*_Y\right|_\alpha ),
  \end{equation}
  where
  \begin{align*}
    \phi_* \left.\ad^*_X\right|_\alpha &= \left.\frac{d}{dt} \phi \left(\Ad^*_{\exp(t X)} \alpha\right) \right|_{t=0}\\
                                       &= \left.\frac{d}{dt} \tau^* \left(\Ad^*_{\exp(t X)} \alpha\right) \right|_{t=0}\\
                                       &= \left.\frac{d}{dt} \Ad^*_{\tau(\exp(t X))} \tau^* \alpha \right|_{t=0} &\tag{by equation~\eqref{eq:coad-tau}}\\
                                       &= \left.\frac{d}{dt} \Ad^*_{\exp(t \tau_* X)} \tau^* \alpha \right|_{t=0}\\
                                       &=\left.\ad^*_{\tau_* X} \right|_{\tau^*\alpha}.
  \end{align*}
  Therefore,
  \begin{align*}
    (\phi^*\omega)_\alpha (\left.\ad^*_X\right|_\alpha, \left.\ad^*_Y\right|_\alpha) &= \omega_{\tau^*\alpha}(\left.\ad^*_{\tau_* X} \right|_{\tau^*\alpha}, \left.\ad^*_{\tau_* Y} \right|_{\tau^*\alpha})\\
                                                                                      &= \left<\tau^*\alpha, [\tau_*X,\tau_*Y]\right> &\tag{by definition of KKS 2-form}\\
                                                                                      &= \left<\tau^*\alpha, \tau_*[X,Y]\right> &\tag{since $\tau_*$ is a Lie algebra automorphism}\\
                                                                                      &= \left<\alpha, [X,Y]\right>\\
                                                                                      &= \omega_{\mathrm{KKS}}(\left.\ad^*_X\right|_\alpha, \left.\ad^*_Y\right|_\alpha),
  \end{align*}
  so that $\phi^*\omega = \omega_{\mathrm{KKS}}$.
\end{proof}

\section{Coadjoint orbits ($n=3$)}
\label{sec:coadjoint-orbits-n=3}

In this appendix we classify the coadjoint orbits of the connected
Carroll group with $n=3$.  This Carroll group has (at least) two
Casimir elements \cite{MR0192900}.  The generator $H$ is central and hence
itself a Casimir. It defines a linear function on $\g^*$ taking the
value $E$ on $\alpha = (\bj,\k,\p,E)$, which is a constant on
coadjoint orbits.  There is also a second Casimir, namely the
euclidean norm $W^2$ of $W_a = H J_a + \epsilon_{abc} P_bB_c$, which
is a quartic symmetric tensor of $\g$.  This also defines a quartic
polynomial function on $\g^*$, taking the value
$\|E \bj + \p \times \k\|^2$ on $\alpha =(\bj,\k,\p,E)$, which is
again constant on coadjoint orbits.

We focus first on the linear Casimir $H$.  There are two main classes
of coadjoint orbits: those for which $E\neq0$ and those for which
$E=0$.

\subsection{Coadjoint orbits with $E\neq 0$}
\label{sec:coadj-orbits-with-nonzero-energy}

Let $\alpha \in \g^*$ be given by $(\bj, \k, \p, E)$ with $E \neq 0$.
We can act with $g(\mathbb{1},\bv,\ba,0)$ with $\ba = -E^{-1}\k$ and $\bv =
E^{-1}\p$ so bring $\alpha$ to $(\bj', \bzero,\bzero, E)$, where
\begin{equation}
  \bj' =\bj + E^{-1} \p\times \k.
\end{equation}
We distinguish two cases.

\subsubsection{Spinless orbits}
\label{sec:spinless}

This results in the coadjoint orbit of $\alpha =
(\bzero,\bzero,\bzero, E)$.  The stabiliser subgroup consists of
matrices of the form
\begin{equation}
  \begin{pmatrix}
    R & \bzero & \bzero\\ \bzero^T & 1 & s \\ \bzero^T & 0 & 1
  \end{pmatrix}
\end{equation}
which is isomorphic to $\SO(3) \times \RR$.  The coadjoint orbit is
therefore six-dimensional.

\subsubsection{Orbits with nonzero spin}
\label{sec:spin-j}

We may use rotations to bring $\bj + E^{-1} \p \times \k$ to any
desired direction.  This results in the orbit of $\alpha = (\bj,
\bzero,\bzero, E)$ with $\bj = \begin{pmatrix} 0 \\ 0 \\
  j \end{pmatrix}$ with $j>0$.  The stabiliser of $\alpha$ consists of
matrices of the form
\begin{equation}
  \begin{pmatrix}
    R & \bzero & \bzero\\ \bzero^T & 1 & s \\ \bzero^T & 0 & 1
  \end{pmatrix} \qquad\text{with $R \bj = \bj$,}
\end{equation}
so it is isomorphic to $\SO(2) \times \RR$.  The coadjoint orbit is
therefore eight-dimensional.

\subsection{Coadjoint orbits with $E=0$}
\label{sec:coadj-orbits-with-zero-energy}

If $E=0$, the coadjoint action of the generic element $g(R,\bv,\ba,s)$
on $\alpha = (\bj,\k,\p,0)$ becomes $(\bj', \k', \p',0)$ where
\begin{equation}
  \label{eq:coadjoint-rep-n=3-e=0}
  \begin{split}
    \bj' &= R \bj + \bv \times R\k + \ba \times R\p\\
    \k' &= R\k\\
    \p' &= R\p.
  \end{split}
\end{equation}
The value of the quartic Casimir on $\alpha$ is given by $\|\p \times
\k\|^2$ which is clearly seen to be invariant, since $\p' \times \k' =
R\p \times Rk = R(\p \times \k)$ for $R \in \SO(3)$, so that $\|\p'
\times\k'\|^2 = \|\p \times \k\|^2$.  We can therefore distinguish
between two cases: those with $\p \times \k \neq \bzero$ and those
with $\p \times \k = \bzero$.  This latter case says that $\p$ and
$\k$ are collinear.  There are several cases here, depending on
whether $\p$ or $\k$ are zero or not.

\subsubsection{\texorpdfstring{$\p = \k = \bzero$}{p=k=0}}
\label{sec:p-k-eq-0}

In this case, $(\bj,\bzero,\bzero,0)$ is mapped to $(R\bj, \bzero,
\bzero, 0)$, so we have two possibilities:

\begin{itemize}
\item If $\bj = \bzero$, the orbit consists of the point $(\bzero,\bzero,\bzero, 0)$.
\item If $\bj \neq \bzero$, the orbit is a sphere of radius
  $\|\bj\|$ and we may choose $\left( \begin{pmatrix}0 \\ 0 \\
      j \end{pmatrix}, \bzero, \bzero, 0\right)$, with $j>0$, as
  representative point in the orbit.
\end{itemize}

\subsubsection{\texorpdfstring{$\k \neq \bzero$, $\p = \bzero$}{k≠0,p=0}}
\label{sec:k-neq-zero}

Here $(\bj,\k,\bzero,0)$ is sent to
$(R\bj + \bv \times R \k,R\k,\bzero,0)$.  We may bring $\k$ to
$\begin{pmatrix}0 \\ 0 \\ k \end{pmatrix}$, with $k>0$, leaving still
the possibility of acting with any $R$ in the stabiliser of $\k$.
Choosing $\bv$ suitably we may bring $\bj$ to $\begin{pmatrix}0 \\ 0
  \\ j \end{pmatrix}$, where $j\in \RR$. In other words, we may take
as a representative of the orbit,
\begin{equation}
  \left( \begin{pmatrix} 0 \\ 0 \\ j \end{pmatrix},
    \begin{pmatrix} 0 \\ 0 \\ k \end{pmatrix}, \bzero, 0\right) \qquad\text{where $j\in \RR$ and $k > 0$,}
\end{equation}
whose stabiliser is the subgroup
consisting of matrices of the form
\begin{equation}
  \begin{pmatrix}
    R & \bzero & \ba \\ \tfrac{v}{k}\k^T & 1 & s + \tfrac12
    \tfrac{v}{k} \ba^T \k\\ \bzero^T & 0 & 1
  \end{pmatrix} \qquad\text{with $R\k = \k$.}
\end{equation}
The stabiliser is six-dimensional, so the orbit is four-dimensional.

\subsubsection{\texorpdfstring{$\p \neq \bzero$, $\k = \bzero$}{p≠0,k=0}}
\label{sec:p-neq-zero}

The story here is very similar to the previous case and, indeed, the
orbits are related via outer automorphisms.  Now $(\bj,\bzero,\p,0)$
is sent to $(R\bj + \ba \times R \p, \bzero, R\p, 0)$.  Hence we may
bring $\p$ to $\begin{pmatrix}0 \\ 0 \\ p \end{pmatrix}$, with $p>0$,
leaving still the possibility of acting with any $R$ in the stabiliser
of $\p$.  By choosing $\ba$ suitably, we may bring $\bj$ to
$\begin{pmatrix}0 \\ 0 \\ j \end{pmatrix}$, where $j\in \RR$.  In
other words, we may take as a representative of the orbit,
\begin{equation}
  \left( \begin{pmatrix} 0 \\ 0 \\ j \end{pmatrix},
    \bzero,  \begin{pmatrix} 0 \\ 0 \\ p \end{pmatrix}, 0\right) \qquad\text{where $j\in \RR$ and $p > 0$,}
\end{equation}
whose stabiliser is the subgroup
consisting of matrices of the form
\begin{equation}
  \begin{pmatrix}
    R & \bzero & \tfrac{a}{p} \p \\ \bv^T & 1 & s + \tfrac12
    \tfrac{a}{p} \bv^T \p\\ \bzero^T & 0 & 1
  \end{pmatrix} \qquad\text{with $R\p = \p$.}
\end{equation}
The stabiliser is six-dimensional, so the orbit is four-dimensional.

\subsubsection{\texorpdfstring{$\p \neq \bzero$, $\k \neq  \bzero$, $\p\times\k=\bzero$}{p≠0,k≠0,p x k=0}}
\label{sec:p-k-neq-zero}

This case is also related to the previous two by automorphisms.  Since
$\p \times \k  = \bzero$, they are collinear: either parallel or
antiparallel.  This means that, letting $\|\p\|=p >0$ and $\|\k\| = k >
0$, $\p \cdot \k = \pm pk$, so the angle between them is $0$ (for the
plus sign) or $\pi$ for the minus sign.  Hence $(\bj,\k,\p,0)$ is sent
to $(R\bj + (\ba \pm \frac{k}{p}\bv) \times R\p, R\k, R\p, 0)$.  We may bring $\p$ to $\begin{pmatrix}0 \\
  0 \\ p \end{pmatrix}$, with $p>0$ and hence $\k$ to $\begin{pmatrix}0 \\
  0 \\ \pm k \end{pmatrix}$.  With $R$ such that $R \p = \p$ (and hence also $R\k =
\k$), we have that $(\bj,\k,\p,0)$ is sent to $(R\bj + (\ba \pm \frac{k}{p} \bv)
\times \p, \k, \p, 0)$ and we may use $\ba \pm \frac{k}{p} \bv$ suitably in
order to make $\bj = \begin{pmatrix}0 \\  0 \\ j \end{pmatrix}$, for
some $j \in \RR$.  In summary, as representative of the orbit we may
take
\begin{equation}
  \left( \begin{pmatrix} 0 \\ 0 \\ j \end{pmatrix},
   \begin{pmatrix} 0 \\ 0 \\ \pm k \end{pmatrix},   \begin{pmatrix}
     0 \\ 0 \\ p \end{pmatrix}, 0\right) \qquad\text{where $p>0$,
   $k>0$ and $j\in \RR$,}
\end{equation}
whose stabiliser is the subgroup consisting of matrices of the form
\begin{equation}
  \begin{pmatrix}
    R & \bzero & \ba \\ \bv^T & 1 & s + \tfrac12
    \ba^T \bv\\ \bzero^T & 0 & 1
  \end{pmatrix} \qquad\text{with $R\p = \p$ and $(\ba \pm \frac{k}{p} \bv)
    \times \p = 0$.}
\end{equation}
The stabiliser is six-dimensional, so the orbit is four-dimensional.

\subsubsection{\texorpdfstring{$\p \times \k \neq \bzero$}{p x k ≠ 0}}
\label{sec:p-times-k-1}

Here $\p$ and $\k$ span a plane, which we can choose to be the plane
of vectors whose first entry is zero.  In other words, under
\begin{equation}
  (\bj, \k,\p,0) \mapsto (R\bj + \ba \times R\p + \bv \times R\k, R\k, R\p,0)
\end{equation}
we may bring $\p$ to $\begin{pmatrix} 0 \\ 0 \\ p \end{pmatrix}$ with
$p>0$ and then with $R$ such that $R\p = \p$, we may bring $\k$  to
$\begin{pmatrix} 0 \\ k_2 \\ k_3 \end{pmatrix}$, with $k_2 \neq 0$.
This fixes the rotational symmetry completely and we now have $(\bj, \k, \p, 0)$ is sent to
$(\bj + \ba \times \p + \bv \times \k, \k,\p,0)$.  It is not hard to
see that by choosing $\ba$ and $\bv$ suitably, we can set $\bj =
\bzero$.  In other words, as a representative of the orbit we may take
the covector
\begin{equation}
  \left( \bzero, \begin{pmatrix} 0 \\ k_2 \\ k_3 \end{pmatrix},   \begin{pmatrix}
     0 \\ 0 \\ p \end{pmatrix}, 0\right) \qquad\text{where $p>0$, $k_3
   \in \RR$ and $k_2 \neq 0$,}
\end{equation}
whose stabiliser is the four-dimensional subgroup of the Carroll group
consisting of matrices of the form
\begin{equation}
  \begin{pmatrix}
    \mathbb{1} & \bzero & \ba \\ \bv^T & 1 & s + \tfrac12 \bv^T \ba\\
    \bzero^T & 0 & 1
  \end{pmatrix} \qquad\text{with $\ba \times \p + \bv \times \k  = \bzero$.}
\end{equation}
The orbit is therefore six-dimensional.

These orbits are listed (listed using the names of the corresponding
Carroll particles) in Table~\ref{tab:carrollions}.

\subsection{Coadjoint orbits of the full Carroll group}
\label{sec:coadj-orbits-full}

The full Carroll group has two connected components, since now $R \in
\Ort(3)$.  The group $\Ort(3)$ is generated by $\SO(3)$ and parity
$P$, which we can think of as space inversion $P\x = - \x$ while
leaving $t$ inert.  We expect that under parity, some coadjoint orbits
are mapped to themselves whereas some other orbits are paired.
However it is clear from the explicit form of the orbit
representatives, that we can always undo the effect of parity
\begin{equation}
  (\bj, \k, \p, E) \mapsto (-\bj, -\k, -\p, E)
\end{equation}
with a rotation, since always $\bj,\k,\p$ have at least one zero
component.  Therefore, the above classification is also the
classification of coadjoint orbits of the full Carroll group.

\subsection{Coadjoint orbits up to automorphisms}
\label{sec:orbits-mod-autos}

Let us now see how automorphisms of the Carroll group relate different
coadjoint orbits.  As we saw in Section~\ref{sec:automorphisms}, the
automorphisms of the Carroll algebra which fix the rotational
subalgebra are given by $\begin{pmatrix} a & b \\ c & d \end{pmatrix}
\in \GL(2,\RR)$ acting via equation~\eqref{eq:autos-carroll}.  The
action on $\alpha \in \g^*$ is via the transpose inverse, so that
$(\bj,\k,\p,E) \mapsto (\bj', \k', \p', E')$ with
\begin{equation}
  \label{eq:autos-momenta}
  \begin{split}
    \bj' &= \bj\\
    \k' &= \frac{d \k - b \p}{ad -bc}\\
    \p' &= \frac{a\p - c \k}{ad - bc}\\
    E' &= \frac{E}{ad-bc}.
  \end{split}
\end{equation}
In particular, we see that both $E$ and $\p \times \k$ transform as
densities of weight $-1$; that is, via multiplication by the
reciprocal of the determinant.

It is now a simple matter to go through the coadjoint orbits listed in
Table~\ref{tab:carrollions} and see how the automorphisms act:
\begin{itemize}
\item[($1$)] Since automorphisms allow us to rescale the energy by a
  nonzero number, all orbits are equivalent under automorphisms, so we can choose
  $(\bzero,\bzero,\bzero,1)$ as orbit representative.  The collection
  of all orbits of nonzero energy and $\bj = \bzero$ can be obtained
  by the combined action of the Carroll group (via the coadjoint
  representation) and the automorphisms from
  $(\bzero,\bzero,\bzero,1)$.
\item[($2$)] We may again rescale the energy to any desired (nonzero) value,
  but we cannot change the spin $S>0$. So we can take $(S \bu,
  \bzero,\bzero, 1)$ as the orbit representative.
\item[($3$)] There is only one orbit of this type.
\item[($4$)] Orbits of this type belong to a one-parameter family labelled by
  $j>0$: automorphisms cannot change $j$.
\item[($5,6,7_\pm$)] The parameter $h$ cannot be changed, but the
  automorphisms act transitively on the other parameters.  Therefore
  up to automorphisms, we have a one-parameter ($h \in \RR$) family of
  orbits.
\item[($8$)] All orbits of this type are equivalent under automorphisms:
  $\p$ and $\k$ span a plane and $\GL(2,\RR)$ acts transitively on the
  bases of that plane.
\end{itemize}
In summary, we may list the equivalence classes of coadjoint orbits of
the Carroll group up to the action of automorphisms as in
Table~\ref{tab:orbits-mod-autos}.

\begin{table}
  \centering
  \caption{Coadjoint orbits up to automorphisms}
  \setlength{\extrarowheight}{3pt}
  \begin{tabular}{>{$}l<{$}>{$}l<{$}>{$}c<{$}}
    \toprule
    \multicolumn{1}{l}{\#} & \multicolumn{1}{c}{Orbit representative} & \dim\mathcal{O}_\alpha\\
                           & \multicolumn{1}{c}{$\alpha=(\bj, \k, \p,E)$} & \\ \midrule \rowcolor{blue!7}
    \relax [1]  & (\bzero, \bzero, \bzero,1) & 6 \\ 
    \relax [2]_S & (S\bu, \bzero,  \bzero , 1)  & 8 \\ \midrule \rowcolor{blue!7}
    \relax [3]  & (\bzero,\bzero,\bzero,0)  & 0  \\     
    \relax [4]_j  & (j\bu,\bzero,\bzero,0) & 2 \\ \rowcolor{blue!7}
    \relax [5,6,7_\pm]_h  & ( h \bu, \bzero, \bu, 0) & 4 \\
    \relax [8] &  (\bzero, \bu, \bu_\perp,0)  & 6\\
    \bottomrule
  \end{tabular}
  \caption*{This table lists the equivalence classes of coadjoint
    orbits of the Carroll group up to the action of automorphisms.
    Some of the parameters in Table~\ref{tab:carrollions} become
    ineffective (up to automorphisms), whereas the greatest
    simplifications come from the fact that all the four-dimensional
    orbits (with the same value of the ``helicity'' $h$) are
    related by automorphisms.  As in Table~\ref{tab:carrollions},
    $\bu$ stands for a fixed, but arbitrary, unit vector and
    $\bu_\perp$ a second unit vector perpendicular to $\bu$.  We see
    that the dimension of the orbit almost determines the class (up to
    automorphisms), except that we have two distinct classes of
    six-dimensional orbits: one with $E \neq 0$ and one with $E = 0$.}
  \label{tab:orbits-mod-autos}
\end{table}

\subsection{Structure of the coadjoint orbits}
\label{sec:struct-coadj-orbits}

The Carroll group is a semidirect product $K \ltimes T$ where $K$ is
the subgroup generated by $J_a,B_a$ and is isomorphic to the
three-dimensional euclidean group, and $T$ is the abelian normal
subgroup generated by $P_a, H$.  In terms of our explicit matrix
parametrisations,
\begin{equation}
  K = \left\{
    \begin{pmatrix}
      R & \bzero & \bzero \\ \bv^T R & 1 & 0 \\ \bzero^T & 0 & 1
    \end{pmatrix}
~ \middle | ~  R \in \SO(3), \bv \in \RR^3\right\}
\end{equation}
 and
\begin{equation}
  T = \left\{
    \begin{pmatrix}
      \mathbb{1} & \bzero & \ba \\ \bzero^T & 1 & s \\ \bzero^T & 0 & 1
    \end{pmatrix}
~ \middle | ~  s \in \RR , \ba \in \RR^3\right\}.
\end{equation}

\subsubsection{Coadjoint orbits of semi-direct products}
\label{sec:coadj-orbits-semi}

Coadjoint orbits of such a semidirect product $K \ltimes T$ are easy
to describe geometrically.  (See, e.g., Oblak's thesis
\cite{Oblak:2016eij}.)  We recall here the main points.  Every
$\alpha \in \g^*$, with $\g = \fk \ltimes \t$, decomposes into
$\alpha = (\kappa,\tau)$ with $\tau \in \t^*$ and $\kappa \in \fk^*$.
The $G$-coadjoint orbit of $\alpha \in \g^*$ fibers over the $K$-orbit
$\eO_\tau$ of $\tau \in \t^*$ under the $K$-action given by the
semidirect product structure.

The action of $K$ on $T$ is by matrix conjugation in $G$:
\begin{equation}
  \begin{pmatrix}
      R & \bzero & \bzero \\ \bv^T R & 1 & 0 \\ \bzero^T & 0 & 1
    \end{pmatrix}
    \begin{pmatrix}
      \mathbb{1} & \bzero & \ba \\ \bzero^T & 1 & s \\ \bzero^T & 0 & 1
    \end{pmatrix}
  \begin{pmatrix}
      R^T & \bzero & \bzero \\ -\bv^T & 1 & 0 \\ \bzero^T & 0 & 1
    \end{pmatrix}
    =
    \begin{pmatrix}
      \mathbb{1} & \bzero & R\ba \\ \bzero^T & 1 & s + \bv^T R \ba\\
      \bzero^T & 0 & 1
    \end{pmatrix},
\end{equation}
from where we can read off the adjoint action on the Lie algebra $\t$
of $T$:
\begin{equation}
  \Ad_{
    \begin{pmatrix}
      R & \bzero & \bzero \\ \bv^T R & 1 & 0 \\ \bzero^T & 0 & 1
    \end{pmatrix}} \begin{pmatrix}
      0 & \bzero & \ba \\ \bzero^T & 0 & c \\ \bzero^T & 0 & 0
    \end{pmatrix} =
    \begin{pmatrix}
      0 & \bzero & R\ba \\ \bzero^T & 0 & c+\bv^T R \ba \\ \bzero^T & 0 & 0
    \end{pmatrix},
  \end{equation}
  or in abbreviated form
\begin{equation}
  \Ad_{(R,\bv)} (\ba, c) = (R\ba , c + \bv^T R \ba).
\end{equation}
The inverse of $(R,\bv)$ is given by $(R^T,-R^T\bv)$, and hence
\begin{equation}
  \Ad_{(R,\bv)^{-1}} (\ba, c) = (R^T\ba , c - \bv^T \ba).
\end{equation}
From this we can work out the action on $\t^*$:
\begin{equation}
  \Ad_{(R,\bv)} (\p, E) = (R\p - E\bv, E).
\end{equation}

Inspecting the coadjoint orbits in
Sections~\ref{sec:coadj-orbits-with-nonzero-energy} and
\ref{sec:coadj-orbits-with-zero-energy}, we see that the
representative $\tau = (\p, E) \in \t^*$ are of three kinds:
\begin{itemize}
\item $\tau = (\bzero,E_0)$, $E_0\neq 0$, whose orbit under $K$ is the  affine hyperplane $\AA^3$ defined by $E = E_0$ in $\t^*$;
\item $\tau = (\bzero,0)$, whose orbit is only that point; and
\item $\tau = (\p\neq\bzero,0)$, whose orbit is the sphere $S^2_{\|\p\|}$ of radius $\|\p\|$ in the hyperplane $E=0$ in $\t^*$.
\end{itemize}

Let $K_\tau \subset K$ denote the stabiliser of $\tau \in \t^*$ and
$\fk_\tau$ denote its Lie algebra.  Let $\kappa_\tau \in \fk_\tau^*$
denote the restriction of $\kappa$ to $\fk_\tau$.  The other
ingredient in the $G$-coadjoint orbit of $\alpha$ is the
$K_\tau$-coadjoint orbit of $\kappa_\tau$, which we denote
$\eO_{\kappa_\tau}$.

Now, the $K$-orbit $\eO_\tau$ of $\tau \in \g^*$ is $K$-equivariantly
diffeomorphic to $K/K_\tau$ and any space on which $K_\tau$ acts,
e.g., $\eO_{\kappa_\tau}$, defines an associated fibre bundle $K
\times_{K_\tau} \eO_{\kappa_\tau} \to \eO_\tau$.  Finally, the
$G$-coadjoint orbit $\eO_\alpha$ of $\alpha = (\kappa,\tau)$ is the
fibred product of $K\times_{K_\tau} \eO_{\kappa_\tau} \to \eO_\tau$
with the cotangent bundle $T^*\eO_\tau \to \eO_\tau$:
\begin{equation}
  \begin{tikzcd}
    \eO_\alpha \arrow[r] \arrow[d] & T^*\eO_\tau \arrow[d] \\
    K\times_{K_\tau} \eO_{\kappa_\tau} \arrow[r] & \eO_\tau\\
  \end{tikzcd}
\end{equation}
The standard notation for this fibred product is
\begin{equation}
  \eO_\alpha = T^*\eO_\tau\times_{\eO_\tau} (K \times_{K_\tau} \eO_{\kappa_\tau}).
\end{equation}
Let us show how to calculate the dimension from this expression.  The
idea is to add the dimensions of the spaces that appear ``above'' (here
$T^*\eO_\tau$, $K$, $\eO_{\kappa_\tau}$) and subtract the dimensions
of the spaces which appear ``below'' (here $\eO_\tau$ and $K_\tau$).
Doing so we obtain
\begin{equation}
  \begin{split}
    \dim\eO_\alpha &= \dim T^*\eO_\tau - \dim \eO_\tau + \dim K - \dim K_\tau + \dim \eO_{\kappa_\tau}\\
    &= 2 \dim \eO_\tau + \dim \eO_{\kappa_\tau}.
  \end{split}
\end{equation}

The cotangent bundle $T^*\eO_\tau$ is itself isomorphic to an
associated vector bundle $T^*\eO_\tau \cong K \times_{K_\tau}
\fk_\tau^0$, where the annihilator $\fk_\tau^0\subset \fk^*$ is
isomorphic to the dual $(\fk/\fk_\tau)^*$.  Hence the coadjoint orbit
$\eO_\alpha$ can also be written as
\begin{equation}
  \eO_\alpha = K \times_{K_\tau} \left( \fk_\tau^0 \times \eO_{\kappa_\tau} \right).
\end{equation}
Again we can calculate the dimension as we did above
\begin{equation}
  \begin{split}
    \dim \eO_\alpha &= \dim K - \dim K_\tau + \dim \fk_\tau^0 + \dim \eO_{\kappa_\tau}\\
    &= \dim K - \dim K_\tau + (\dim \fk - \dim \fk_\tau) + \dim
    \eO_{\kappa_\tau}\\
    &= 2 \dim \eO_\tau + \dim \eO_{\kappa_\tau},
  \end{split}
\end{equation}
resulting (of course) in the same dimension.

\subsubsection{Coadjoint orbits of the euclidean groups}
\label{sec:coadj-orbits-eucl}

To see how this works in practice and because we shall need these
results below, let us discuss the coadjoint orbits of the euclidean
groups $\ISO(n)$.  We are particularly interested in $n=2,3$ as those
appear as stabilisers $K_\tau$ in our discussion of Carroll orbits.

The euclidean group $\ISO(n)$ is a subgroup of the affine group
$\Aff(n,\RR)$ and hence we may embed it as a subgroup of the general
linear group $\GL(n+1,\RR)$.  We choose the following embedding
\begin{equation}
  \ISO(n) = \left\{
    \begin{pmatrix}
      R & \bv \\ \bzero^T & 1
    \end{pmatrix}
~ \middle | ~  R \in \SO(n), \bv \in \RR^n \right\}.
\end{equation}
We shall use the abbreviated notation $(R,\bv)$ for the generic element
of $\ISO(n)$.  Then the group law in this parametrisation is given by
\begin{equation}
  (R_1, \bv_1) (R_2, \bv_2) = (R_1 R_2, \bv_1 + R_1 \bv_2),
\end{equation}
from where we read off the inverse
\begin{equation}
  (R,\bv)^{-1} = (R^T, - R^T\bv),
\end{equation}
where we have used that $R^{-1} = R^T$ for $R \in \SO(n)$.  Let us
introduce the following abbreviated notation for the Lie algebra
$\iso(n)$:
\begin{equation}
  (X,\bb) =
  \begin{pmatrix}
    X & \bb\\ \bzero^T & 0
  \end{pmatrix}.
\end{equation}
The adjoint representation is easy to work out explicitly and one
finds
\begin{equation}
  \Ad_{(R,\bv)^{-1}} (X,\bb) = (R^T X R, R^T(\bb + X \bv)).
\end{equation}
This allows us to write the coadjoint action.  Write $(J,\k) \in
\iso(n)^*$ with dual pairing
\begin{equation}
  \left<(J,\k), (X,\bb)\right> = \tfrac12 \Tr J^T X + \k^T\bb.
\end{equation}
Then it follows from a simple calculation that
\begin{equation}
  \Ad^*_{(R,\bv)}(J,\k) = (RJR^T + (R\k)^T\bv - \bv (R\k)^T, R\k).
\end{equation}

Let us now specialise to $n=2$.  Since $\SO(2)$ is abelian, we have
now that $RJR^T = J$, so that
\begin{equation}
  \Ad^*_{(R,\bv)}(J,\k) = (J + (R\k)^T\bv - \bv (R\k)^T, R\k).
\end{equation}

We see that the norm $\|\k\|$ is an invariant of the orbit.  If $\k =
\bzero$, the orbit of $(J,\bzero)$ is a point $\{(J,\bzero)\}$.

If $\k \neq\bzero$, we may bring it to the form
$\k =\begin{pmatrix} 0 \\ k\end{pmatrix}$, where $k = \|\k\| > 0$ and
then using $\bv$ we may set $J = 0$.  Therefore the orbit
representative is $(0, (0,k)^T)$ and the stabiliser consists of
matrices of the form
\begin{equation}
  \begin{pmatrix}
    1 & 0 & 0\\
    0 & 1 & w\\
    0 & 0 & 1
  \end{pmatrix}
\end{equation}
which is abelian and hence has point-like coadjoint orbits.  The
$\SO(2)$-orbit of $\k$ is a circle of radius $\|\k\|$ and hence the
$\ISO(2)$-coadjoint orbit of $(0,\k)$ is the cotangent bundle
$T^*S^1_{\|\k\|} \cong S^1_{\|k\|} \times \RR$.

In other words, the $\ISO(2)$-coadjoint orbits are either cylinders of
some positive radius $\|\k\|$, or every point on the line $\k =
\bzero$.

Finally, let us specialise to $n=3$.  Now a generic element of
$\iso(3)^*$ is $(\bj,\k)$ which is sent to $(R\bj + \bv \times R\k,
R\k)$ under the coadjoint action of $(R,\bv)$.  Again $\|\k\|$ is
invariant.

If $\k=\bzero$, $(\bj,\bzero)$ is sent to $(R\bj,\bzero)$,
so that $\|\bj\|$ is invariant.  If $\bj = \bzero$, we have a
point-like orbit $\{(\bzero,\bzero)\}$; otherwise we have the sphere
$S^2_{\|\bj\|}$ of radius $\|\bj\|$.

If $\k \neq 0$, we may rotate $\k = (0,0,k)^T$, where $k =
\|\k\|>0$.  Rotating in such a way that we leave $\k$ fixed, we can
bring $\bj = (0,j_2,j_3)^T$, but then by choosing $\bv$ suitably, we
can bring $\bj = (0,0,j)$ for some $j\in\RR$.  In other words, the
orbit representative is
\begin{equation}
  \left(
    \begin{pmatrix}
      0 \\ 0 \\ j
    \end{pmatrix},
    \begin{pmatrix}
      0 \\ 0 \\ k
    \end{pmatrix}
\right) \qquad\text{with $j\in\RR$ and $k>0$.}
\end{equation}
The stabiliser of this $(\bj,\k)$ is the two-dimensional abelian group
consisting of elements $(R,\bv)$ with $R\k = \k$ and $\bv \times \k =
\bzero$.  Since it is abelian, its coadjoint orbits are points.  In
summary, the coadjoint orbits of $(\bj,\k)$ with $\k\neq \bzero$ are
of the form
\begin{equation}
  \left\{ \begin{pmatrix} 0 \\ 0 \\ j \end{pmatrix}\right\} \times T^*S^2_{\|\k\|}.
\end{equation}

\subsubsection{Structure of the coadjoint orbits of the Carroll group}
\label{sec:struct-coadj-orbits-1}

We now determine the structure of the Carroll coadjoint orbits.  It is
now a simple matter to go through each of the coadjoint orbits and
identify $\kappa$, $\tau$, $K_\tau$ and $\kappa_\tau$ and the
$K_\tau$-coadjoint orbit $\eO_{\kappa_\tau}$ of $\kappa_\tau$.
Table~\ref{tab:structure-orbits} summarises these results, with
explanations following the table.

\begin{table}[h]
  \centering
    \caption{Deconstructing the coadjoint orbits}
    \label{tab:structure-orbits}
  \begin{adjustbox}{max width=\textwidth}
    \begin{tabular}{>{$}l<{$}|*{8}{>{$}c<{$}}}
      \# & \alpha \in \g^* & \tau \in \t^* & \eO_\tau & K_\tau & \kappa \in \fk^* & \kappa_\tau \in \fk_\tau^* & \eO_{\kappa_\tau} & \eO_\alpha\\\toprule
      1& (\bzero,\bzero,\bzero,E_0\neq0) & (\bzero,E_0) & \AA^3_{E=E_0} & \SO(3) & (\bzero,\bzero) & \bzero &  \{\bzero\} & T^*\AA^3 \\
      2& (\bj\neq\bzero,\bzero,\bzero,E_0\neq0) & (\bzero,E_0) & \AA^3_{E=E_0} & \SO(3) & (\bj,\bzero) & \bj &  S^2_{\|j\|} & T^*\AA^3 \times_{\AA^3} (K \times_{K_\tau} S^2) \\
      3& (\bzero,\bzero,\bzero,0) & (\bzero,0) & \{(\bzero,0)\} & K & (\bzero,\bzero) & (\bzero,\bzero) &  \{(\bzero,\bzero)\} & \{(\bzero,\bzero,\bzero,0)\}\\
      4& (\bj\neq\bzero,\bzero,\bzero,0) & (\bzero,0) & \{(\bzero,0)\} & K & (\bj,\bzero) & (\bj,\bzero) &  S^2_{\|\bj\|} & S^2 \\
      5& (\bj,\k\neq\bzero,\bzero,0)_{\bj\times\k=\bzero} & (\bzero,0) & \{(\bzero,0)\} & K & (\bj,\k) & (\bj,\k) & \{\bj\} \times T^*S^2_{\|\k\|} & T^*S^2\\
      6& (\bj,\bzero,\p\neq 0,0)_{\bj\times \p = \bzero} & (\p,0) & S^2_{\|\p\|} & \SO(2) \ltimes \RR^3 & (\bj,\bzero) & (\bj,\bzero) & \{(\bj,\bzero)\} & T^*S^2\\
      7& (\bj,\k\neq\bzero,\p\neq\bzero,0)_{\k \times \p = \bj \times \k = \bzero} & (\p,0) & S^2_{\|\p\|} & \SO(2) \ltimes \RR^3 & (\bj,\k) & (\bj,\k) & \{(\bj,\k)\} & T^*S^2\\
      8& (\bzero,\k,\p,0)_{\k\times\p \neq \bzero} & (\p,0) & S^2_{\|\p\|} & \SO(2) \ltimes \RR^3 & (\bzero,\k) & (\bzero,\k) & T^*S^1_{\|\k\|} & T^*S^2 \times_{S^2} (K \times_{K_\tau} T^*S^1)\\
      \bottomrule
    \end{tabular}
  \end{adjustbox}
  \vspace{1em}
  \caption*{The stabilisers $K_\tau$ in cases \#6,7,8 consists of
    elements $(R,\bv)$ where $\bv\in\RR^3$ is arbitrary and $R\p = \p$.
    Since $\p \neq \bzero$, these are rotations about the axis defined
    by $\p$ and hence isomorphic to $\SO(2)$, so that the stabiliser
    is isomorphic to $\SO(2)\ltimes \RR^3$, but actually $\SO(2)$
    only acts nontrivially on a plane in $\RR^3$, hence the stabiliser
    is more properly written as $\ISO(2) \times \RR$ or, even more
    invariantly, as $\ISO(\p^\perp) \times \RR \p$, as a
    subgroup of $\ISO(3)$.  The $K_\tau$-coadjoint orbits should be
    self-explanatory.  In cases \#3,4,5, the stabiliser is $K \cong
    \ISO(3)$, whose coadjoint orbits were determined in
   Section~\ref{sec:coadj-orbits-eucl} above: in case \#3 we have the
    point-like orbit $\{\bzero,\bzero\}$, in case \#4 we have the
    2-sphere of radius $\|\bj\|$, and in case \#5 we have the
    cotangent bundle of the sphere of radius $\|\k\|$.  In cases
    \#6,7,8, the stabiliser $K_\tau$ is isomorphic to $\ISO(2)\times
    \RR$ and the coadjoint orbits have been determined in
    Section~\ref{sec:coadj-orbits-eucl} above.  In case \#6 we have
    the point-like orbit $\{(\bj,\bzero)\}$ and in case \#8 we have
    the cylinder $T^*S^1$.  Only case \#7 needs some explanation.
    In this case, all of $\bj,\k,\p$ are collinear.  The
    $K_\tau$-coadjoint action of $(R,\bv)$ on $\kappa_\tau = (\bj,\k)$
    gives $(\bj +  \bv \times \k, \k)$, but $\bj + \bv \times \k \in
    \so(3)$ \emph{and} we need to project to $\so(2)$.  This projection is
    along $\bj$ and since $(\bv \times \k)\cdot \bj = 0$, we see that
    the coadjoint orbit is only the point $\{(\bj,\k)\}$.}
\end{table}

\providecommand{\href}[2]{#2}\begingroup\raggedright\endgroup

\end{document}